\documentclass[aps, reprint, nofootinbib, showkeys, amsmath, amssymb, superscriptaddress, floatfix]{revtex4-2}
\makeatletter
\def\Dated@name{}
\makeatother
\usepackage{float}
\usepackage{amsthm}
\usepackage{amsfonts}
\usepackage{braket}
\usepackage{soul}
\usepackage{simpler-wick}
\usepackage{ wasysym }
\usepackage{placeins}
\usepackage{soul}
\usepackage{color}
\usepackage{comment}
\usepackage[dvipsnames]{xcolor}
\usepackage{listings}
\usepackage{amsmath, amssymb}
\usepackage[utf8]{inputenc}
\usepackage{graphicx}
\usepackage{booktabs}
\usepackage{enumerate}
\usepackage[font={small},labelfont={bf}]{caption}
\usepackage{multirow}
\usepackage{subfigure}
\usepackage{url}
\usepackage{caption}
\usepackage{subcaption}
\usepackage[figurename=Figure]{caption}
\usepackage{bbm} 
\usepackage{physics}
\newcommand{\pZTL}[1]{{\fontfamily{lmss}\selectfont pZTL}}
\newcommand{\ZTL}[1]{{\fontfamily{lmss}\selectfont ZTL}}
\theoremstyle{plain} 
\newtheorem{theorem}{Theorem}[section]
\newtheorem{definition}{Definition}[section]
\newtheorem{corollary}{Corollary}[section]
\newtheorem{proposition}{Proposition}[section]
\newtheorem{lemma}{Lemma}[section]
\newcommand{\ED}[1]{\textcolor{orange}{Ellen: #1}}

\usepackage[pdftex, pdftitle={Article}, pdfauthor={Author}]{hyperref}

\begin{document}

\bibliographystyle{apsrev4-2}

\title{Exponential improvement in benchmarking multiphoton interference}

\author{Rodrigo M. Sanz}
\email{rmarsan3@upv.edu.es}
\affiliation{Universitat Politècnica de València, Camino de Vera s/n, 46022 Valencia, Spain}
\author{Emilio Annoni}
\email{emilio.annoni@quandela.com}
\affiliation{ Centre for Nanosciences and Nanotechnology, Université Paris-Saclay, UMR 9001,10 Boulevard Thomas Gobert, 91120, Palaiseau, France}
\affiliation{Quandela, 7 Rue Léonard de Vinci, 91300 Massy, France}
\author{Stephen C. Wein}
\affiliation{Quandela, 7 Rue Léonard de Vinci, 91300 Massy, France}
\author{Carmen G. Almudéver}
\affiliation{Universitat Politècnica de València, Camino de Vera s/n, 46022 Valencia, Spain}
\author{Shane Mansfield}
\author{Ellen Derbyshire}
\email{ellen.derbyshire@quandela.com}
\author{Rawad Mezher}
\email{rawad.mezher@quandela.com}
\affiliation{Quandela, 7 Rue Léonard de Vinci, 91300 Massy, France}

\begin{abstract}
Several photonic quantum technologies rely on the ability to generate multiple indistinguishable photons. Benchmarking these photons' level of indistinguishability is essential for scalability. The Hong–Ou–Mandel dip provides a benchmark for the indistinguishability between two photons,  and extending this test to the multi-photon setting has so far resulted in a protocol that computes the genuine $n$-photon indistinguishability (GI). However, this protocol has a sample complexity increasing exponentially with the number of input photons for an estimation of GI up to a given additive error. 

To address this problem, we first introduce new theorems that strengthen our understanding of the relationship between distinguishability and the suppression laws of the quantum Fourier transform  interferometer (QFT). Through this, we propose  a protocol using the QFT for benchmarking GI, and which  achieves constant sample complexity for the estimation of GI to a  given additive error for prime photon numbers,  and sub-polynomial scaling otherwise; an exponential improvement over previous state-of-the art. We prove the optimality of our protocol in many relevant scenarios and validate our approach experimentally on Quandela's reconfigurable photonic quantum processor, where we observe a clear advantage in runtime and precision over the state-of-the-art. We therefore establish the first scalable method for  computing  multi-photon indistinguishability, and which applies naturally to current and near-term photonic quantum hardware.

\end{abstract}


\maketitle

\section{Introduction} \label{sec:intro}
Photonic quantum technology has made significant strides in recent years, driven by advances in single-photon sources \cite{dingHighefficiency2025, thomasBright2021, alexanderManufacturable2025, somaschi2016near} and integrated optical architectures \cite{maringVersatile2024, zhongQuantum2020, pelucchiPotential2022}. Crucial to the development of this promising technology is the generation and manipulation of multiple \emph{perfectly} \emph{indistinguishable} photons. Partial distinguishability can fundamentally undermine the usefulness of a multi-photon state~\cite{renemaEfficient2018, renemaClassical2019} and introduces a barrier to creating quantum interference. However, state-of-the-art single-photon sources still produce particles that are, to some degree, distinguishable. Therefore, understanding~\cite{menssenDistinguishability2017, tichyManyparticle2012, tichySampling2015, tichyZeroTransmission2010, shchesnovichPartial2015} and characterizing~\cite{springChipbased2017, pontQuantifying2022, rodariSemiDeviceIndependent2025} the physics of imperfect multi-photon interference is an essential and ongoing area of study. To continue scaling up photonic quantum computers without compromising on their computing power, it is important to develop an efficient and reliable way to benchmark multi-photon indistinguishability.

The reference benchmark for indistinguishability is the Hong-Ou-Mandel (HOM) dip~\cite{hongMeasurement1987}. In this work, it is established that the level of indistinguishability between two photons can be determined from a binary interference pattern. However, this has been shown to be insufficient for characterizing multi-photon interference~\cite{menssenDistinguishability2017}. Scalable techniques \cite{brodWitnessing2019} have been developed that give upper and lower bounds on the probability  of $n$ photons being perfectly indistinguishable,  a metric known as \emph{genuine $n$-photon indistinguishability (GI)}~\cite{menssenDistinguishability2017}. However, up until now, the only  method we are aware of for \emph{exactly} computing GI was that introduced by Pont et al.~\cite{pontQuantifying2022}. Their method relies on a cyclic integrated interferometer (CI) and a post-selection strategy that homogenizes the output probability for any (partially) distinguishable input state as soon as there is at least one distinguishable photon at input. Thus, allowing for a discrimination between two scenarios: fully indistinguishable input and anything else. This creates a multi-photon HOM-like interference pattern for their post-selected output configurations allowing them to extract GI. The drawback of this method is that the post-selected outputs are rarely observed, leading to a measurement protocol with a sample complexity scaling exponentially as $O(4^{n})$ for estimating GI up to a given fixed  additive error.

In this work we introduce a  protocol for measuring GI based on a quantum Fourier transform (QFT) interferometer. Our protocol builds on a significant body of work exploring the relationship between the QFT and distinguishability~\cite{tichyZeroTransmission2010, dittelTotally2018, saiedGeneral2024,menssenDistinguishability2017, bezerraFamilies2023}. For fully indistinguishable inputs with the QFT, certain output configurations are suppressed, dictated by the zero-transmission laws (ZTL)~\cite{tichyZeroTransmission2010, tichyManyparticle2012}. We present several new theorems based on  the key insight that for partially distinguishable input states these ``suppressed'' output configurations are populated with a uniform probability. By post-selecting on these outcomes, we obtain a direct estimator for GI with a protocol that is \emph{exponentially more efficient} than the one based on the CI. Indeed, we show that  the sample complexity of our protocol is $O(1)$, and is \emph{optimal},  for estimating GI to a given fixed additive precision for prime $n$. For all other $n$, the same sample complexity scales as $f(n)$, where $f(n)$ is sub-polynomially small in $n$.

We confirm our theoretical conclusions with an experimental comparison between a QFT and CI implementation on one of Quandela's photonic quantum processors~\cite{maringVersatile2024}, for $n=3$ and $n=4$ photons. Remarkably, we find that even at these small system sizes and without photon number resolving detectors, our protocol has a shorter running time and provides a more precise measurement of GI. We therefore envision that our protocol will become a standard tool for benchmarking multi-photon interference.

\section{Background}
\label{sec: background}

Consider the following $n$-photon input state 

\begin{equation}\label{eq: dist_state}
    \rho = c_1\rho^\parallel + \sum\limits_k c_k \rho_k^\perp\enspace ,
\end{equation}
with $c_1 + \sum_k c_k = 1$, where $\rho^\parallel$ is the fully indistinguishable state. The fully indistinguishable state is weighted by a probability $c_1$, which is exactly equal to the genuine $n$-photon indistinguishability (GI). The states $\rho^\perp_k$ are pure states called  partition states \cite{annoniIncoherent2025} (see also \cite{brodWitnessing2019, giordani2020experimental}). Note that for distinct states  $\sigma_1, \sigma_2 \in \{\rho^{\parallel}, \{\rho_k^{\perp}\} \}$,  $\mathsf{Trace}(\sigma_1, \sigma_2)=0$. Partition states, a term introduced by Annoni and Wein~\cite{annoniIncoherent2025}, are $n$-photon states that are \emph{partitioned} into perfectly distinguishable subsets, containing one or more photons that are indistinguishable within each subset. 

While a generic $n$-photon state need not follow the above decomposition~\cite{menssenDistinguishability2017, shchesnovichCollective2018}, nevertheless it is always possible to \emph{twirl} such a  state into the form of Eq.\ref{eq: dist_state} by performing mode permutations of the state, and averaging over these~\cite{annoniIncoherent2025}. This is akin to the Pauli twirling procedure used in the qubit gate model, where coherent error channels are twirled into incoherent Pauli error channels \cite{wallmanNoise2016}.

Pont et al.~\cite{pontQuantifying2022}
proposed an interferometer, called the cyclic  interferometer (CI),  to estimate $c_1$. Their setup ensures that, for a specific set of outputs, all $\rho_k^\perp$ as input states will yield identical output probabilities. The protocol therefore \emph{post-selects} on these outputs, leading to a post-selected probability of the form 
\begin{equation}\label{eq: discrim_probability}
    P = c_1P_i + (1-c_1)P_d \enspace ,
\end{equation}

where $P_i$ is the output probability of the post-selected events given a fully indistinguishable input,  and $P_d$  is the output probability of the post-selected events for any $\rho^{\perp}_k$ at input. Their choice of post-selected outputs results in a probability of $P=\frac{1}{2^n}(1+\theta c_1),$ where $\theta \in [-1,1]$ is an adjustable parameter of the circuit. From the factor $\frac{1}{2^n}$ it should be clear that a reliable estimation of $c_1$ (to within a fixed additive error) requires an estimation of $P$ to an exponentially small (in $n)$ precision. By invoking standard statistical arguments \cite{hoeffdingProbability1963}, it is easy to show that an exponential number of runs of the above  protocol are needed to reliably estimate $c_1$. This exponential scaling thus hinders the scalability of this method beyond a few photons.

In  our work, we  seek a similar protocol to that of~\cite{pontQuantifying2022} in terms of post-selecting on specific outputs of a well-designed interferometer. However, in addition to that we require our protocol to be \emph{efficient}, that is we require the number of samples needed to estimate $c_1$ to a fixed precision to scale \emph{at most} as a very slow growing function of $n$.  

More precisely, we will look for an interferometer $U$ and a set of post-selected outputs $\mathcal{S}$ that satisfy Eq.~\ref{eq: discrim_probability}, with the constraint that $P_i=0$, and $P_d$ is maximized.  Viewed in this way, our problem can be cast as finding an \emph{optimal} state discrimination protocol~\cite{barnettQuantum2009}, implementable with linear optics, between $\rho^{\parallel}$  and $\{\rho^{\perp}_k\}$, $\forall k$.

A particular case of this type of discrimination task was considered by Stanisic and Turner~\cite{stanisicDiscriminating2018}, where the goal was to discriminate between $\rho^{\parallel}$ and a partition state $\rho^s$ with exactly one photon distinguishable from the remaining ($n-1$) indistinguishable photons. The maximum probability, known as the \emph{optimal} success probability, achievable with any set of measurement settings (post-selected outputs) and any interferometer and  for $\rho^s$ at input was shown \cite{stanisicDiscriminating2018} to be $1-\frac{1}{n}$ achieved when using post-selected output probabilities of a QFT. In our work, we will build upon this result, and on subsequent works, highlighting the importance of the zero-transmission laws (ZTLs)~\cite{tichyManyparticle2012, tichyZeroTransmission2010,saiedGeneral2024} (see Section~\ref{sec: ZTLs_main}) for fully indistinguishable inputs.

\subsection{Setting}
We represent any pure $n$-photon state in $m$ modes as a vector $\vec{s} = (s_1, ..., s_m)$, with $s_j$ the number of photons in mode $j$. 

A \emph{partition state} can be further decomposed as a sum of  vectors $\vec{s}^{\ (i)}$ representing the internal states of the subsets, we refer to these as \emph{registers} indexed with $i$, such that
\begin{equation}
    \vec{s} := \sum_i \vec{s}^{\ (i)} \enspace ,
\end{equation}
where $\vec{s}^{\ (i)} := (s_1^{(i)}, ..., s_m^{(i)})$ and $ \sum_i s_j^{(i)} = s_j$. Photons in each register $i$ are fully indistinguishable among themselves, and fully distinguishable from photons in register $j \neq i$.

An important class of partition states are \emph{$t$-periodic partition states} which are decomposed as
\begin{equation}
    \vec{s}_t := \sum_i \vec{s}_{t_i}^{\ (i)} \enspace ,
\end{equation}
where $\vec{s}_{t_i}^{\ (i)}$ is a state in register $i$ that has a repeating pattern of length $t_i$ of single indistinguishable photons in each mode, e.g. $\vec{s}_{t_i = 2}^{(i)} = (1_i, 0, 1_i, 0)$. The total period of the vector $\vec{s}_t$ is $t= \mathsf{lcm}(t_1, \dots,t_K)$, $K$ being the number of registers.

Finally, we define the \emph{mode-assignment vector} which, for any state vector $\vec{s}$, is:
\begin{equation}
    d(\vec{s}) := (d_1 (\vec{s}), ..., d_n (\vec{s})) \enspace ,
\end{equation}
where $d_i(\vec{s}) \in \{1, ..., m\}$ with $d_i (\vec{s}) \leq d_{i+1} (\vec{s})$ labelling the modes of each photon. For example, let $\vec{s} = (2, 1, 0, 1)$, then $ \ d(\vec{s}) = (1, 1, 2, 4)$.

As a final example, consider the case where $m=n$ and exactly one photon occupies each mode. This specific setting will be used later on to derive our results. In this case, a fully indistinguishable state, where a pattern of one photon in a single register is repeated $n$ times, has $t=1$.  In contrast,  a fully distinguishable state, where there is no repeated pattern across any register, is non-periodic and  has a period $t=n$. Indeed, and by  convention, we will say a non-periodic state always has a period of $t=m$.

\subsection{Zero-transmission Laws}\label{sec: ZTLs_main}
Tichy et al.~\cite{tichyManyparticle2012, tichyZeroTransmission2010} introduced suppression laws for a fully indistinguishable $n$-photon state passing through a multi-port balanced beamsplitter. The corresponding $m \times m$ unitary matrix is the QFT, defined as:
\begin{equation}\label{eq: QFT_def}
    QFT_m = \frac{1}{\sqrt{m}}\begin{bmatrix}
        1 & 1 & \cdots & 1 \\
        1 & \omega^1 & \cdots & \omega^{m-1}\\
        1 & \omega^2 & \cdots & \omega^{2(m-1)} \\
        \vdots & \vdots & \ddots & \vdots \\
        1 & \omega^{m-1} & \cdots & \omega^{(m-1)(m-1)}
    \end{bmatrix} \enspace ,
\end{equation}

where $\omega=e^{2 \pi i/m}$. 

The suppression laws, known as the zero-transmission laws (ZTL), are dictated by the \emph{Q-value} of an output configuration $\vec{s}$
\begin{equation}\label{Q_value}
    Q(\vec{s}): = \mod\left(\sum\limits_i d_i(\vec{s}), m\right) \enspace ,
\end{equation}
where $\sum\limits_i d_i(\vec{s})$ is the sum of the elements of vector $d(\vec{s})$. 

Given an input state $\vec{r}$ passed through a QFT interferometer, we denote the probability of obtaining an output state $\vec{s}$ as $P_{\vec{r} \to \vec{s}}$. The ZTLs~\cite{saiedGeneral2024} are then defined as follows.

\pZTL{} (\emph{periodic ZTL}). For any $t$-periodic input state of fully indistinguishable photons $\vec{r}_t$:
\begin{equation}\label{eq: pZTL}
    Q(\vec{s}) \neq 0 \ \text{mod} \ \left(\frac{m}{t}\right) \implies P_{\vec{r}_t \to \vec{s}} = 0 \enspace.
\end{equation}
This means that the output configuration $\vec{s}$ will be strictly suppressed if the above is true. 

A more suitable form of the ZTL for our purposes is the following.

\ZTL{}. For an $n=m$ fully indistinguishable input state with one photon per mode:
\begin{equation}\label{eq: ZTL}
    Q(\vec{s}) \neq 0  \implies P_{\vec{r}_{t=1} \to \vec{s}} = 0 \enspace,
\end{equation}
where we have denoted as $\vec{r}_{t=1}$ the input state of $n$ indistinguishable photons, which has period $t=1$, as explained in the example of the previous section.

\subsection{Transition probability}
Our results rely on proving uniform patterns in the transition probability for our proposed set-up. Let $U \in \mathbb{C}^{m \times m}$ represent the scattering matrix of a linear optical interferometer and let $d(\vec{r}), d(\vec{s}) \in [n]^m$ be the mode-assignment vectors for an input and output photon state, respectively. The transition probability from input $\vec{r}$ to output $\vec{s}$ is recovered from the effective scattering matrix, which is a sub-matrix of $U$
\begin{equation}
\label{eq: M_matrix}
    M_{d(\vec{r}), d(\vec{s})} \in \mathbb{C}^{n\times n} 
    \enspace ,
\end{equation}
formed by first taking $\vec{r}_i$ copies of the $i_{th}$ column of $U$ and then taking $\vec{s}_j$ copies of the $j_{th}$ row constructed in the previous step~\cite{aaronsonComputational2011}.

One of our contributions is to derive the transition probability for a partition state with $K$ distinguishable registers, using the formalism of Tichy et al.~\cite{tichySampling2015}, we show that 
\begin{equation}
    \label{eq: part_probability}
    P_{\vec{r} \to \vec{s}}= \sum\limits_{\{\vec{s}^{(j)}\} \in \mathcal{S}}  \left(\prod_{j= 1}^K N_j|\mathsf{Per}(M_{d(\vec{r}^{(j)}), d(\vec{s}^{(j)})})|^2\right) \enspace ,
\end{equation}
where the normalization factor for each register is $N_j = 1/\prod_i r_i^{(j)} ! s_i^{(j)}!$ and accounts for multiple occupancy of modes. $\mathsf{Per}(.)$ is the  matrix permanent, which is used in computing output probabilities of indistinguishable photons passed through an interferometer \cite{aaronsonComputational2011}.  Here, we sum over all combinations of register outputs $\vec{s}^{\ (j)}$ that result in output $\vec{s}$, denoted by the set $\mathcal{S}$. This is derived in detail in the Supplementary Material (S.M.), Section~\ref{sec: transition_prob_derivation}. Intuitively, Eq.~\ref{eq: part_probability} shows that we can treat the output probabilities of partition states as products of output probabilities of sets of indistinguishable photons defined by the different registers. 

\section{Main Results} \label{sec:results}

We propose an efficient and scalable protocol for computing $c_1$, the GI, based on the output probabilities of the QFT interferometer. Our input is an $n$-photon state $\rho$ of the form of Eq.~\ref{eq: dist_state} that is passed through a QFT interferometer with mode number $m=n$. The input state is such that each mode is occupied by exactly one photon. Following intuition laid out in Section~\ref{sec: background}, we propose to post-select outcomes that \emph{would} be suppressed by the \ZTL{} \emph{if} the input state were indistinguishable. In order to do this, we require photon-number-resolving (PNR) detectors. By introducing several new theorems, we prove that with our protocol, we can estimate the value of $c_1$ up to \emph{any} additive error $\epsilon$ with high confidence. 

When $n$ is prime, and the desired precision $\epsilon$ is fixed, we show that our method achieves an \emph{optimal} sample complexity of $O(1-\frac{1}{n})$ simplifying to $O(1)$ as $n$ grows. The overall sample complexity for computing $c_1$ to \emph{any} additive error $\epsilon$ is therefore $O(\frac{1}{\epsilon^2})$, as determined by Hoeffding's inequality \cite{fisherCollected1994}.

For arbitrary $n$, and again with fixed precision $\epsilon$, we show that there is a classical postprocessing protocol for estimating $c_1$, that runs in time $O(n)$ and has sample complexity $o(n)$. Consequently, the total sample complexity to estimate $c_1$ within error $\epsilon$ is then $o(\frac{n}{\epsilon^2})$.

The setup for our protocol is illustrated in Figure~\ref{fig: protocol}.

\begin{figure}[H]
\centering    \includegraphics[width=0.9\columnwidth]{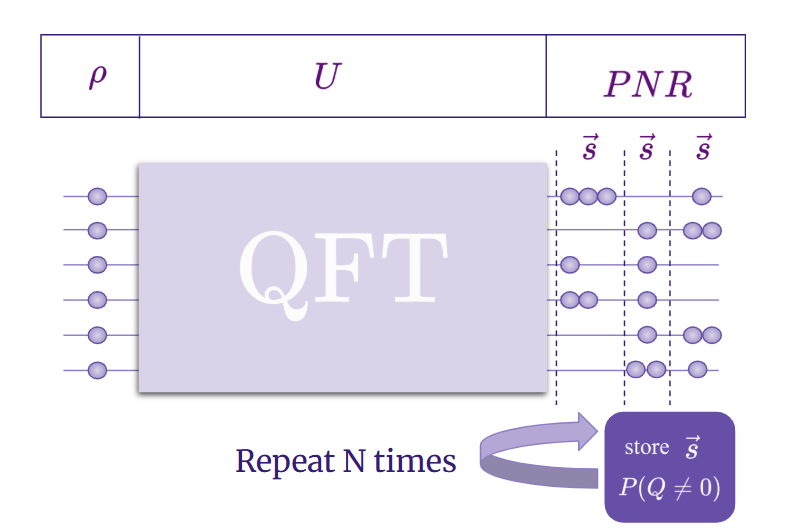}
    \caption{Illustration of the protocol. An $n$-photon state is input to an $n$-mode QFT interferometer, and certain output events are post-selected on using photon-number-resolving (PNR) detection. The estimated output probabilities from these events are then used to estimate $c_1$.}
    \label{fig: protocol}
\end{figure}

\subsection{Computing relevant  probabilities}
In order to establish our results, we first analyze the behavior of the probabilities required for our protocol. 
An important component of any distinguishability analysis is the distinguishability matrix. Given an $n$-photon state, we denote two photons in this state by single-photon pure states $|r_k^{(i)}\rangle$ and $|r_l^{(j)}\rangle$ in registers $i,j$.  The distinguishability matrix~\cite{tichyManyparticle2012} of such a state, has elements defined as 
\begin{equation}
    \mathrm{S}_{k,l} = \langle r_k^{(i)}| r_l^{(j)}\rangle = \delta_{i,j}\enspace .
\end{equation} 
Where, for two indistinguishable photons the element $\mathrm{S}_{k,l} = 1$ and zero otherwise.

We denote the probability of observing outputs such that $Q=k$ with $$P(Q = k):=\sum_{\vec{s} \big| Q(\vec{s})=k} P_{\vec{r} \to \vec{s}} \enspace .$$ 

One of our  crucial  findings is that the probability $P(Q=k)$ is actually dictated by the \emph{periodicity} of the input partition state. More precisely, we prove the following theorem.

\begin{theorem}[Q-marginals probability]\label{thm: q_marginals_prob_main}
    Let $\sigma \in \{\rho^{\parallel}, \rho^{\perp}_k\}$ be an $n$-photon partition state.
    The Q-values produced by $\sigma$ when passed through an $n$-mode QFT interferometer with exactly one input photon per mode are 
    \begin{equation}
    \label{eq: q_marginals_main}
        P(Q=k) = \frac{1}{n}\sum_{j=0}^{n-1}\omega^{kj} \left[ \prod\limits_{i=1}^n \mathrm{S}_{i, i+j}\right], 
    \end{equation}
    where $\mathrm{S}_{i, i+j}$ are elements of the distinguishability matrix.
\end{theorem}

The proof of this theorem relies firstly on expanding $P(Q=k)$ as a sum of permanents corresponding to the post-selected output probabilities. Then, we show that each of these permanents reduces to an $n$th root of unity. Finally, we show that only \emph{shift permutations} that map photon $i$ to $i+j \mod n$, for a fixed $j$, contribute to $P(Q=k)$. Together, these reduce $P(Q=k)$ to the sum in Eq.~\ref{eq: q_marginals_main}. The full proof is detailed in S.M. Section~\ref{sec: probabilities}

This result enables us to show that the probability for $t$-periodic partition states becomes
\begin{equation}\label{eq: Q-marginals-t}
    P(Q=k) = \begin{cases}
        \frac{1}{t} \ \text{if} \ k=0 \ \text{mod} \ \frac{n}{t} \\ 
        0 \  \text{otherwise}
    \end{cases} \enspace .
\end{equation}
This displays the uniform behavior of output probabilities when input states are of the same periodicity, which we prove in S.M. Section~\ref{sec: partition_results}.

\subsection{Computing $c_1$ for prime $n$}

To estimate $c_1$, we first look at the case when $n$ is prime.

We reformulate our problem in terms of our post-selection strategy. By applying the projector $O=\sum_{s \big | Q(\vec{s}) \neq 0} |s\rangle \langle s|$ to the $n$-photon state in Eq.~\ref{eq: dist_state}, we find that
 
\begin{equation}
\label{eqinitQ}
    P(Q \neq 0)=c_1P_i(Q \neq 0)+\sum_kc_kP_k( Q \neq 0).
\end{equation}

As in the case of the CI \cite{pontQuantifying2022}, our goal is to achieve the following simple form for the post-selected probability

\begin{equation}\label{eq: prob_Q_main}
    P(Q\neq0)=c_1 P_{i}(Q\neq0) + (1-c_1)P_{d}(Q\neq0) \enspace .
\end{equation}

Meaning that we require $\forall k$ that we have the same probability, i.e. $P_{k} ( Q \neq 0)=P_d( Q \neq 0)$. Crucially, and in order to improve exponentially over the performance of the CI, we need $P_d( Q \neq 0)$ to be sufficiently large. For the case when $n$ is prime we prove that Eq. \ref{eq: prob_Q_main} holds. Furthermore, in this case we prove that $P_d(Q \neq 0)$ is significantly large and in fact its magnitude is optimal (see Section \ref{sec:opt}). Namely, we show that    

\begin{equation}\label{eq: analytic_prob}
\begin{split}
    &P_i(Q\neq 0) = 0, \\ 
    &P_d(Q\neq 0) = 1 - 1/n  \enspace .
\end{split}
\end{equation}

We prove the above analytical probabilities in S.M. Section~\ref{sec: partition_results}, using Eq.~\ref{eq: Q-marginals-t} and the fact that $t\in\{1, n\}$ when $n$ is prime.

With these results we can show that  $P(Q\neq0)=(1-c_1)(1-1/n)$. Thus,  we can compute $c_1$ by determining the post-selected probability $P(Q\neq0)$ experimentally and substituting in our probabilities from Eq.~\ref{eq: analytic_prob} to get
\begin{equation}\label{eq: compute_c1}
    c_1 = \frac{(1-1/n)-P(Q\neq0)}{1-1/n} \enspace .
\end{equation}

\subsection{Computing $c_1$ for  $n$ non-prime}\label{sec: non-prime-main}

When $n$ is non-prime, the partition states can have several possible periods $t \in \{1, \dots, n\}$. It is therefore necessary to first re-write Eq.~\ref{eq: dist_state} taking into account all of these different periods \footnote{Note that because $m=n$ and we have exactly one photon per mode in the input, the set of possible periods  is equal to the  set of divisors of $n$.}
\begin{equation}\label{eq: non-prime-rho}
     \rho = c_{t_1}\cdot \rho_{t_1} + \sum_j c_{j,t_2}\cdot \rho_{j,t_2} + ...+\sum_jc_{j,t_x(n)}\cdot  \rho_{j,t_x(n)} \enspace ,
\end{equation}
where $t_1, ..., t_{x(n)}$ are the $x(n)$ divisors of $n$, $t_1=1$ and $t_{x(n)}=n$. Whilst $\rho_{j,t_i}$ are partition states with period $t_i$. Here, $\rho_{t_1}$ is the fully indistinguishable state, with periodicity $t_1=1$, and $c_{t_1}=c_1$ (GI). 

Applying Eq.~\ref{eq: Q-marginals-t} which tells us that all states of the same period $t$ have equal probability to the above reformulated state, we obtain the following expression for the post-selected probability
\begin{equation}
\label{eq: non-prime-Q-main}
\begin{split}
    P(Q=k) &= P_{t_1}(Q=k)\cdot c_{t_1} + P_{t_2}(Q=k)
    \cdot c_{t_2}\\
    &+...+ P_{t_{x(n)}}(Q=k)\cdot  c_{t_{x(n)}}  \enspace ,
\end{split}
\end{equation}
where $c_{t_i}:=\sum_jc_{j,t_i}$. 

In contrast to Eq.~\ref{eq: prob_Q_main} where $P(Q\neq 0)$ is a convex combination of two probabilities, we can see that for any $k$ we will have a mixture over multiple components $P_{t_i}(Q=k)$. To manage this,   we replace the $P_{t_i}(Q=k)$ with their values according to Eq.~\ref{eq: Q-marginals-t}, and construct the following system of equations 

$$ A \vec{c}=\vec{P} \enspace ,$$

where $A=(a_{i,j})$ is a real-valued $n \times x(n)$ matrix with matrix elements defined by Eq.~\ref{eq: Q-marginals-t}
\begin{equation}
    a_{i,j} = P_{t_j}(Q=k) = \frac{1}{t_j}\delta_{i \ \text{mod} \frac{n}{t_j}} \enspace ,
\end{equation}
where $\delta_{i \text{ mod } \frac{n}{t_j}}=1$ if $  i= 0 \text{ mod } \frac{n}{t_j}$, and 0 otherwise,   $\vec{P}=(P(Q=1), \dots , P(Q=n))^T$ is a column-vector of the post-selected probabilities,  and 
$\vec{c}=(c_{t_1}, \dots, c_{t_{x(n)}})^T$ is a vector of  the coefficients.

Our strategy will inevitably be different from that of the prime case. We show that the columns of $A$ are linearly independent, which ensures that the system can be inverted using the Moore-Penrose pseudo-inverse~\cite{mooreReciprocal1920, penroseGeneralized1955} $A^+$. Notably, we also recover an analytic form for $A^+$ allowing us to recover $c_1$ by  using the experimental estimates of $P(Q=k)$ for all $k$. These results are derived in detail in S.M. Section~\ref{sec: non-prime_partition_results}.

\bigskip
An important aspect of our analysis is the sample complexity, which we discuss for both the prime and non-prime cases in the next section.

\subsection{Scaling results}

\subsubsection{Prime scaling is optimal}
\label{sec:opt}

 We have seen in the previous sections that we can compute $c_1$, for general partition states when $n$ is prime  using Eq.~\ref{eq: compute_c1}. In  these cases, the benchmarking task reduces to a discrimination problem between just two output distributions, where the success probability is given by $P_s=P_d-P_i=1-1/n$. Notably, this results in an asymptotic sampling cost scaling as $O(\frac{1}{P_s^2}) \in O(1)$, as seen in S.M. Section~\ref{sec: prime-scaling}. This means that the number of required shots $N_s$ to estimate $c_1$ with a fixed statistical error and  confidence remains independent of $n$. Such behavior is illustrated in Figure~\ref{fig: all_partitions_scaling}, where one sees that the lower bound for $N_s$ saturates at $\sim 1200$ samples. Note  that to plot this figure we required an estimation $\tilde{c_1}$ of $c_1$ such that  $P\!\left(\,|\tilde{c}_1 - c_1| \le \epsilon\,\right) \ge 0.995, \epsilon = 0.05$. Using the results of Stanisic and Turner~\cite{stanisicDiscriminating2018} we show that the scaling for prime $n$ is 
 
 in fact optimal.

\begin{theorem}[GI success probability]\label{thm: gi_optimal}
Let $n \in \mathbb{N}^+$, let $\{U, \mathcal{S}\}$ be an interferometer paired with a post-selection strategy that measures $c_1$.
Then $P(s\in \mathcal{S}) \leq 1-1/n$.
\end{theorem}

We prove this by considering the worst-case mixture of Eq.~\ref{eq: dist_state} and argue by contradiction. The details of this proof are in S.M. Section~\ref{sec: optimality}

In fact, the optimality of the success probability directly implies optimality of the sample complexity, since by Hoeffding's inequality \cite{hoeffdingProbability1963} we have that $N_s \in O(\frac{1}{P_s^2})$ for an estimate with a fixed additive error.

\subsubsection{Non-prime scaling}

In the non-prime case, we first assume that each estimate of $P(Q=k)$ is accurate up to an additive error $\epsilon$. By propagating this uncertainty through the application of $A^+$ one can derive that the error $\epsilon'$ on $c_1$ is amplified by a factor containing in its denominator Euler's totient function \cite{hardyIntroduction2008}, denoted as $\varphi(n)$. Indeed, in S.M. Section~\ref{sec: non-prime-scale}, we show that 
\begin{equation}    
\epsilon' \leq \epsilon \frac{n}{\varphi(n)}.
\end{equation}
For computing $c_1$ to a fixed additive error $\delta$, we therefore need $\epsilon = \frac{\delta}{\frac{n}{\varphi(n)}}$ additive error on the probability estimations, and therefore $N_s \in O( \frac{1}{\epsilon^2}) \in O( \frac{n^2} {\varphi(n)^2})$ \cite{hoeffdingProbability1963}. 

\begin{figure}[thb]
\centering    \includegraphics[width=1\columnwidth]{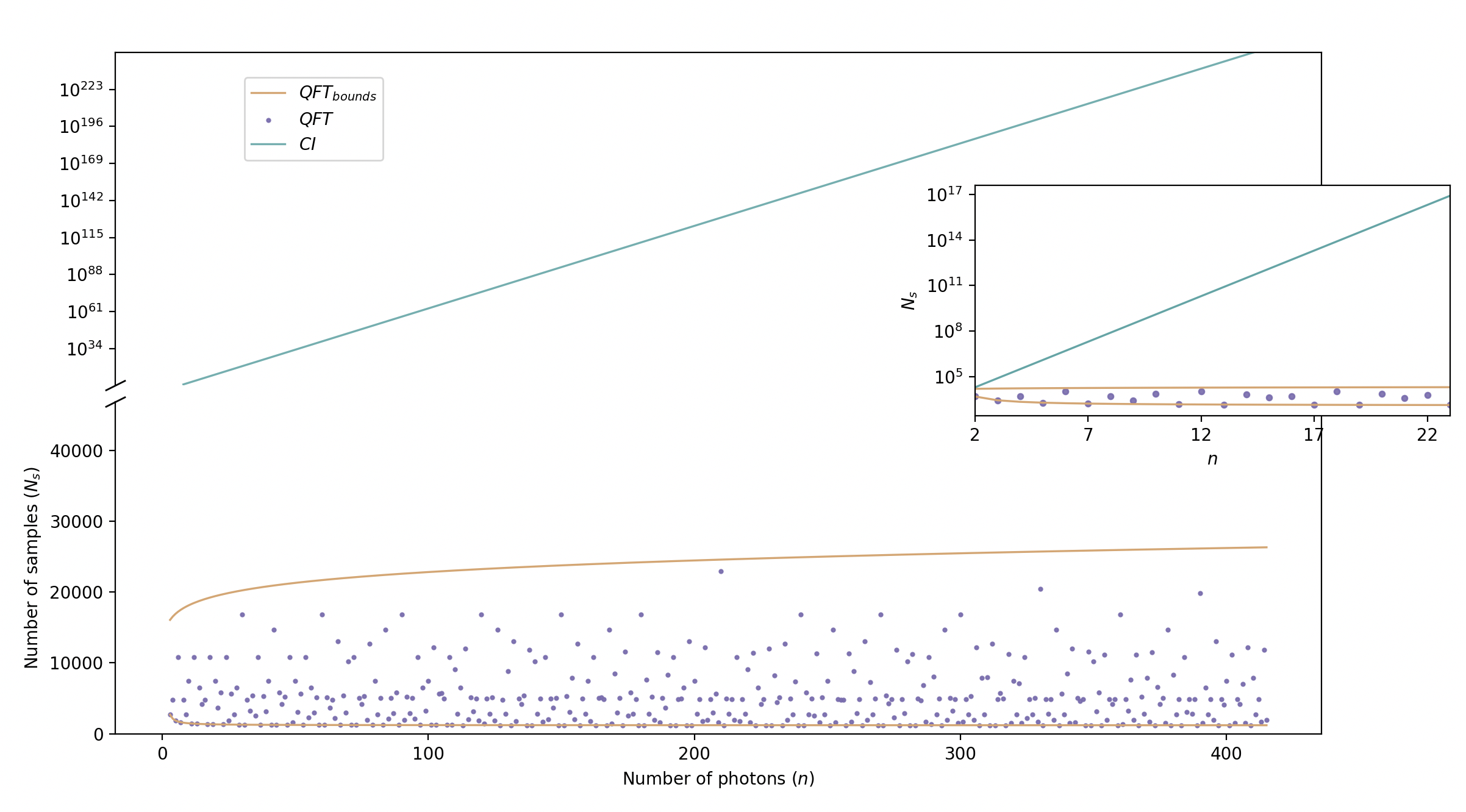}
    \caption{Illustration of the sampling requirements to benchmark genuine $n$-photon indistinguishability (GI) ($c_1$) with fixed additive error and statistical confidence. The cyclic interferometer (CI) exhibits exponential overhead, whereas our QFT-based protocol achieves near-constant complexity for prime $n$ and sub-polynomial growth for non-prime $n$ in the worst case. The number of samples is established using concentration bounds derived from a Hoeffding inequality \cite{hoeffdingProbability1963}. The upper and lower bound functions for the QFT (yellow lines) in the plot are proportional to $n^{\frac{1}{10}}$ and $(\frac{n}{n-1})^2$.}
    \label{fig: all_partitions_scaling}
\end{figure}
Interestingly, and as illustrated in Figure~\ref{fig: all_partitions_scaling}, $\frac{n^2}{\varphi(n)^2}$ equals  the  optimal scaling of $(\frac{n}{n-1})^2$ when $n$ is prime.
For non-prime $n$, $\frac{n}{\varphi(n)}$ is upper bounded  by $x(n)$  which is known to have a subpolynomial scaling \cite{hardyIntroduction2008}. Thus,  $\frac{n^2}{\varphi(n)^2} \in o(n)$.

In addition to the  sample complexity--- typically the most expensive for near-term devices ---our protocol for  non-prime $n$ also has a classical runtime cost: computing and applying the pseudoinverse. This task requires $O(n)$-time, and finding all $x(n)$ divisors of $n$ requires $O(\sqrt{n})$-time \cite{knuthArt1997}. This can be seen in S.M. Section~\ref{sec: non-prime-scale}. 

Table~\ref{tab:all_partitions_scaling} summarizes the sample complexity scalings of our protocol as compared to that of the CI.

\begin{table}[thb]
\centering
\begin{tabular}{lcc}
\toprule
\textbf{Protocol} & \textbf{$n$} & \textbf{Asymptotic scaling of $N_s$} \\
\midrule
CI  & all       & $O(4^n)$ \\
\addlinespace
QFT & prime     & $O(1)$ \\
    & non-prime & $o(n)$ \\
\bottomrule
\end{tabular}
\caption{Asymptotic sample complexity $N_s$ for computing the genuine $n$-photon indistinguishability (GI) ($c_1$). The scalings correspond to general states $\rho$ with support on arbitrary partition states, comparing our QFT-based protocol to the state-of-the-art CI protocol.}\label{tab:all_partitions_scaling}
\end{table}

\section{Experimental implementation}
To support our theoretical results, we implemented our protocol on the \emph{Ascella} photonic quantum processor~\cite{maringVersatile2024}, available on the Quandela cloud. It has a quantum-dot based single-photon source~\cite{thomasBright2021} with time-to-space de-multiplexing to inject up to six photons into a 12-mode universal reconfigurable linear-optical interforemeter. Outputs are monitored by 12 superconducting nanowire single-photon detectors (SNSPD) that can only detect whether photons are present in a mode, rather than how many.

We benchmarked the QFT$_n$ against the CI$_n$ for $n=3,4$ using pseudo-photon-number-resolving detection with integrated $1\times n$ beam splitters. To ensure identical inputs, CI$_n$ was padded with two modes of the Identity. In a single 14 hour run (to minimize noise drift) we obtained $c_1=0.815\pm0.001$ from $QFT_3$, $c_1=0.824\pm0.007$ from $CI_3$, $c_1=0.695\pm0.007$ from $QFT_4$ and $c_1=0.72\pm 0.01$ from $CI_4$. Our results are shown in Figure~\ref{fig:experiments}.

The two methods produced consistent results, though small discrepancies remain which we attribute to slow fluctuations of the processor alignment and periodic recalibration. Another possibility is that the input state slightly deviates from the partition representation. In the future, permutation twirling~\cite{annoniIncoherent2025} can be enforced on the $CI_n$ to ensure that both approaches measure the same value.

Although it is difficult to rigorously compare the efficiency of the two approaches due to the differences in PPNR and post-processing techniques, it is clear that the QFT$_n$ provided a more precise estimate of $c_1$ with a lower total runtime on the processor. This is most evident for $n=3$ where the runtime of the QFT$_3$ was $\sim 32$ minutes, less than a third of that of the CI$_3$ which was $\sim$ 162 minutes, whilst the resulting $c_1$ is 7 times more precise.
\label{sec: experiment}
\begin{figure*}
    \centering
    \includegraphics[width=\linewidth]{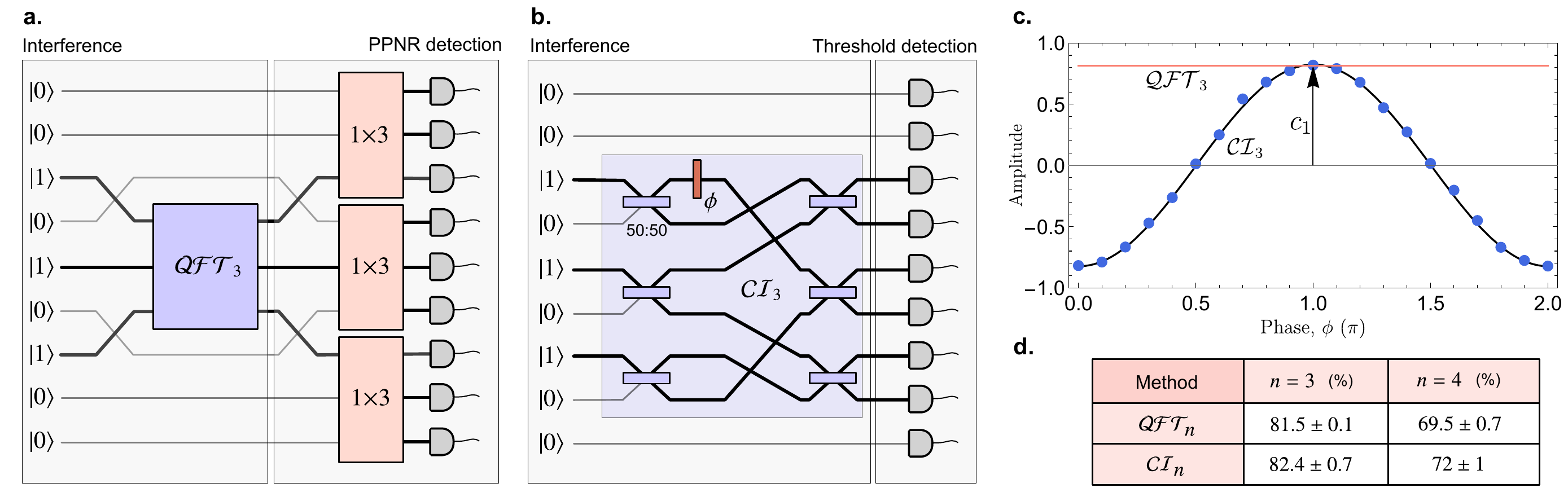}
    \caption{\textbf{Experimental implementation.} \textbf{a.} Interference between three single photons using $QFT_3$ along with pseudo-number resolving (PPNR) detection up to 3-photon resolution. \textbf{b.} Interference between three single photons using a cyclic interferometer ($CI_3$) tuned using a single phase shifter implementing the phase $\phi$. The desired outcomes contain no bunching and so direct threshold detection suffices. Note that the first two modes and the last mode went unused in this experiment, since the interferometer was shifted so that the same three photons were measured for both methods. \textbf{c.} The interference fringe obtained from the output of $CI_3$ whose amplitude of oscillation quantifies the genuine 3-photon indistinguishability $c_1$ (vertical arrow). The horizontal line notated by $QFT_3$ indicates the value of $c_1$ obtained from the experiment illustrated in panel \textbf{a.}. \textbf{d.} A table summarizing the $n=3$ and $n=4$ genuine multi-photon indistinguishabilities measured using the reconfigurable chip for both methods.}
    \label{fig:experiments}
\end{figure*}
\newpage
\section{Discussion}
\label{sec:discussion}

In summary, we have proposed a scalable protocol for directly computing GI, an important benchmark  of multi-photon interference.  
The capability  of our protocol to scale up is clearly reflected in Figure~\ref{fig: all_partitions_scaling}, where from $n=19$ onwards both the prime and non-prime QFT protocols already exhibit a reduction of  orders of magnitude in sampling cost compared to the CI. This substantial advantage highlights the impact of our protocol, which we envision  will become a practical and scalable tool for benchmarking multi-photon indistinguishability in near-term and future photonic platforms.

An interesting open question is proving the optimality of the sample complexity of  our protocol for the non-prime $n$ case. One observation we can make is that the ratio $\frac{n}{\varphi(n)}$ governing the sample complexity is equal to the optimal scaling of $\frac{n}{n-1}$ factor when $n$ is prime. This provides evidence that our protocol may be optimal for any $n$, although further investigation is required to confirm this conjecture. 

Up until now, most classical simulation algorithms for noisy boson sampling experiments have relied on simplified assumptions on the distinguishability error~\cite{renema2019simulability,moylett2019classically}. Our work could aid in developing similar algorithms with generic partition state inputs, which cover a wide variety possible input states and provide a realistic model of distinguishability. 

Finally, an important direction motivated by bringing GI on par with existing qubit benchmarks \cite{eisert2020quantum} is investigating the level of GI required for reliable quantum computing operations, such as  multi-qubit gates implemented in linear optics \cite{knill2001scheme,baldazzi2025universal,maringVersatile2024}, or indeed reliable multi-photon state preparations \cite{aralov2025photon,kopylov2025multiphoton}.

\section{Acknowledgements}
This work has been funded by the HorizonCL4 program under Grant Agreement No. 101135288 for the
EPIQUE project; by the QuantERA project ResourceQ
under Grant Agreement No. ANR-24-QUA2-007-003; and by the TUF-TOPIQC project financed by the French National Quantum Strategy (France 2030) program. RMS and CGA acknowledge support from the Ministry for Digital Transformation and of Civil Service of the Spanish Government through the QUANTUM ENIA project call - Quantum Spain project, and by the European Union through the Recovery, Transformation and Resilience Plan - NextGenerationEU within the framework of the Digital Spain 2026 Agenda.

We became aware of  independent upcoming works  \cite{Novo2026,RobbioInPrep} establishing a method for  multivariate trace estimation, with applications to computing GI.
 
\bibliography{ref_zotero}

\newpage
\onecolumngrid
\appendix 
\section*{Supplementary Material: Exponential improvement in benchmarking multiphoton interference}
\section{Preliminaries}

We write an $n$-photon Fock state, in $m$ modes as~\cite{aaronsonComputational2011, tichyZeroTransmission2010, tichySampling2015}:
\begin{equation}
    |\phi (\vec{r})\rangle = \prod\limits_{i=1}^m \frac{1}{\sqrt{r_i!}}  (\hat{a}_{i}^\dagger)^{r_i}|0\rangle \enspace ,
\end{equation}
where $|0\rangle$ is the vacuum state on $m$ modes and $a_i^\dagger (a_i)$ are the creation (annihilation) operators that follow the bosonic commutation relations:
\begin{equation}
    [\hat{a}_i^, \hat{a}_j^\dagger] = \delta_{i,j} \ , [\hat{a}_i, \hat{a}_j] = [\hat{a}_i^\dagger, \hat{a}_j^\dagger] = 0 \ . 
\end{equation}
A Fock state can always be described by it's vector $\vec{r} = (r_1, ..., r_m)$ where $r_i$ is the number of photons in mode $i$ and $i=1, ..., m$. Throughout our work we denote input state vectors as $\vec{r}$ and output state vectors as $\vec{s}$.

A unitary matrix representing the optical interferometer acts on creation and annihilation operators~\cite{camposQuantummechanical1989}, mapping $\hat{a}_i^\dagger$ to $\hat{b}_j^\dagger$ via:
\begin{equation}\label{eq: unitary_transform}
    \hat{b}_j^\dagger = \sum\limits_{k=1}^m U_{j,k} \hat{a}_k^\dagger \enspace .
\end{equation}

\begin{definition}[t-periodic Fock state]\label{defn: t-periodicFock}
    a Fock state vector of length $m$, with a repeating pattern of length $t$ of configurations of photons in modes, where the photons are fully indistinguishable
\begin{equation}\label{eq: fock_periodic_state}
    \vec{r}_t := (r_1, ..., r_t, r_1, ..., r_t, ..., r_1, ..., r_t) \enspace ,
\end{equation}
where $r_i \in (0, ..., n)$ and we label the number of photons per period $t$ as $n_t$.
\end{definition}
 Examples of $t$-periodic states are: $\vec{r}_{3} = (1, 0, 0, 1, 0, 0)$ or $\vec{r}_{3} = (2, 1, 0, 2, 1, 0)$.
 \\

\begin{definition}[Mode-assignment vector] \label{defn: mode_assignment}for any state vector $\vec{s}$
\begin{equation}
    d(\vec{s}) := (d_1 (\vec{s}), ..., d_n (\vec{s})), d_i(\vec{s}) \in (1, ..., m) \ \text{with} \ d_i (\vec{s}) \leq d_{i+1} (\vec{s})\enspace ,
\end{equation}
labelling the modes of each photon.
\end{definition}
For example, let $\vec{s} = (2, 1, 0, 1)$ then $ \ d(\vec{s}) = (1, 1, 2, 4)$.
\\
\begin{definition}[Q-value]\label{defn: q_value}
for any state vector $\vec{s}$
\begin{equation}\label{Q_value}
    Q(\vec{s}) = \mod\left(\sum\limits_i d_i(\vec{s}), m\right) \enspace ,
\end{equation}
where $\sum\limits_i d_i(\vec{s})$ is the sum of the elements of vector $d(\vec{s})$. 
\end{definition}
Using the above value of $d(\vec{s}) = (1, 1, 2, 4)$, the Q-value is $Q(\vec{s}) = (1 + 1 + 2 + 4) \mod 4 = 0$.

\subsection{Partial distinguishability}
For $n$ photons, we define a partially distinguishable partition state in density matrix form, using the approach of refs.~\cite{brodWitnessing2019, giordaniExperimental2020}:
\begin{equation}\label{eq: partially_distinguishable}
    \rho = c_1\rho^\parallel + \sum_j c_j \rho_j^\perp \enspace ,
\end{equation}
where $\rho^\parallel$ is a pure state of $n$ fully indistinguishable photons and $\rho_j^\perp$, $j < 1$, are pure states consisting of subsets of fully indistinguishable photons, where \emph{at least} two photons are distinguishable (orthogonal) to each other - this includes the fully distinguishable state, in which case the subsets are of size $1$. Note that the $\rho_j^\perp$ represent all possible partition states which we define below.

The value of $c_1$ is the probability of the mixed state being fully indistinguishable and is what we refer to as \emph{genuine n-photon indistinguishability}. \\
\begin{definition}[Partition state]\label{defn: partition_state}
A pure state of $n$ photons in $m$ modes that can be partitioned into $K$ registers of indistinguishable photons that are completely orthogonal to photons in other registers. It can be written in state form as:
    \begin{equation}
        |\phi(\vec{r})\rangle := \left(\prod\limits_{j=1}^K \prod\limits_{i=1}^m \frac{1}{\sqrt{r_i^{(j)}!}} (\hat{a}_{i, |\psi_j\rangle}^\dagger)^{(r_i^{(j)})}\right)
        |0\rangle \enspace ,
    \end{equation} 
where the superscript $^{(j)}$ denotes the registers $j=1, ..., K$ where all photons in register $j$ are indistinguishable from each other, therefore sharing internal degrees of freedom represented by $|\psi_j\rangle$.

The vector form of the partition state can be written as: $\vec{r} = \sum\limits_{j=1}^K \vec{r}^{\ (j)}$, which we visualise below:
\begin{equation}\label{eq: partition_state}
\begin{split}
    \vec{r} = (&r_1, \ \   r_2 \ ,..., \  r_m) \\
     & \ \ \ \ \rotatebox[origin=c]{90}{$=$} \\ 
    \vec{r}^{\ (1)} = (&r_1^{(1)}, r_2^{(1)}, ..., r_m^{(1)})\\
    & \ \ \ \ \rotatebox[origin=c]{90}{$+$} \\
    &\ \vdots \ \ \ \ \ \ \vdots \ \ \ \ \ \ \ \ \ \vdots \\
    & \ \ \ \ \rotatebox[origin=c]{90}{$+$} \\ 
    \vec{r}^{\ (K)}= (&r_1^{(K)}, r_2^{(K)}, ..., r_m^{(K)}) \enspace ,
\end{split}
\end{equation}
where $r_i \in (0, ..., n)$ and $\sum_{j=1}^K r_i^{(j)} = r_i$, the number of photons in mode $i$.
We often refer to the states $|\phi(\vec{r}^{\ (j)})\rangle$ (or vector form $\vec{r}^{\ (j)}$) as register-states.
\end{definition}

We further define a special case of a partition state, which is a simpler model of partial distinguishability, known as the \emph{orthogonal bad bit (OBB)} model. We define the OBB partition state in vector form only, for simplicity. 
\\

\begin{definition}[OBB partition state]\label{defn: obb_state}
a partition state of length $m$, $\vec{r} = \sum_{j=0}^{k} \vec{r}^{\ (j)}$ with one register ($j=0$) containing fully indistinguishable photons and $k$ registers, each containing exactly one photon that is fully distinguishable from all others.

Formally, a $k$-distinguishable OBB partition state can be written as:
\begin{equation}
    \vec{r} = \vec{r}^{\ (0)} + \sum_{j=1}^{k} \vec{r}^{\ (j)},
\end{equation}
where:
\begin{itemize}
    \item $\vec{r}^{\ (0)}$ is a Fock state corresponding to a register of photons that are fully indistinguishable among themselves,
    \item Each $\vec{r}^{\ (j)} = (r_1^{(j)}, ..., r_m^{(j)})$ (see above Eq.~\ref{eq: partition_state}), for $j = 1, \dots, k$, corresponds to a single register containing exactly one photon that is fully distinguishable from all other photons, i.e., $\sum_{i=1}^m r_i^{(j)} = 1$.
\end{itemize}

By construction, no subset of distinguishable photons can share internal degrees of freedom: each photon is either fully indistinguishable from those in $\vec{r}^{\ (0)}$, or fully distinguishable from all others.
\end{definition}

Finally, we define another periodic state, similar to Definition~\ref{defn: t-periodicFock}
\\

\begin{definition}[t-periodic partition state]\label{defn: t-periodic_partition}
a pure state vector of length $m$ with $n=m$ photons, and one photon per mode, which is a union of $t_i$-periodic Fock states $\vec{r}_{t_i}^{\ (i)}$ where the period $t = lcm(t_i)$. 
\end{definition} 
To best illustrate this, consider the following periodic Fock states, labelled by their registers $i=1, 2, 3$ such that each state contains indistinguishable photons, with periodicities $t_1, t_2, t_3$:
\begin{equation}
\begin{split}
    \vec{r}_{t_1}^{\ (1)} &= (1_1, 0, 1_1, 0, 1_1, 0, 1_1, 0) \\
    \vec{r}_{t_2}^{\ (2)} &= (0 , 0, 0, 1_2, 0, 0, 0, 1_2) \\
   \vec{r}_{t_3}^{\ (3)} &= (0, 1_3, 0, 0, 0, 1_3, 0, 0) \enspace .
\end{split}
\end{equation}
 Since we have one photon per mode by design, for every $\vec{r}_{t_i}^{\ (i)} = (r_{1}^{(i)},.., r_{m}^{(i)})$ we denote $r_{j}^{(i)} = 1_i$ for each register $i$ and mode $j$.
 
Then, the union of the above states is a $t$-periodic partition state of the form:
\begin{equation}\label{eq: partition_periodic_state}
\vec{r}_{t} = (1_1, 1_3, 1_1, 1_2 , 1_1, 1_3, 1_1, 1_2)  \enspace ,
\end{equation}
with period $t=\mathrm{lcm}(t_1, t_2, t_3) = \mathrm{lcm}(2, 4, 4) =4$. 

\subsection{Transition probability}

Let $\hat{U}$ be the unitary transformation induced by the linear optical interferometer $U \in \mathbb{C}^{m \times m}$. Defining our input state as $|\phi(\vec{r})\rangle$, the probability of transition to the output state $|\phi(\vec{s})\rangle$ is simply:
\begin{equation}\label{eq: transition_prob}
    P_{\vec{r} \rightarrow \vec{s}} = |\langle \phi(\vec{s})|\hat{U}|\phi(\vec{r})\rangle|^2
\end{equation}

Let $d(\vec{r}), d(\vec{s}) \in [n]^m$ be the mode-assignment vectors for the input and output, respectively, see Definition~\ref{defn: mode_assignment}. Then we define the effective scattering matrix 
describing the transition from $\vec{r}$ to $\vec{s}$ as the sub-matrix of $U$:
\begin{equation}
\label{eq: M_matrix}
    M_{d(\vec{r}), d(\vec{s})} \in \mathbb{C}^{n\times n} 
    \enspace ,
\end{equation}
formed by first taking $\vec{r}_i$ copies of the $i_{th}$ column of $U$ and then taking $\vec{s}_j$ copies of the $j_{th}$ rows of this new matrix~\cite{aaronsonComputational2011}. 

For coincidence input and output states, where maximally one photon appears per mode, Tichy et al.~\cite{tichySampling2015} derive the general transition probability to be:
\begin{equation}
\label{eq: transition_probability}
    P_{\vec{r} \rightarrow \vec{s}} := \sum\limits_{\sigma, \rho \in \mathsf{S}_n} \prod\limits_{j=1}^n M_{\sigma(j), j} M_{\rho(j), j}^* \mathrm{S}_{\rho(j), \sigma(j)} \enspace , 
\end{equation}
where $M_{i,j}^*$ denotes the complex-conjugate of the matrix $M$ with columns (rows) permuted according to $i (j)$ and $\mathsf{S}_n$ is the permutation group of $n$ photons. $\mathrm{S}_{l,k}$ are the elements of the \emph{distinguishability matrix}, which encodes how distinguishable pairs of photons are and is defined as

\begin{equation}
    \mathrm{S}_{k,l} = \langle r_k^{(i)}| r_l^{(j)}\rangle = \delta_{i,j}\enspace,
\end{equation} 
with elements equal to $1$ when the photons ($r_k,r_l$) are indistinguishable, belonging to the same register $(i,j)$, and zero otherwise.

Note that, if the input state is $t$-periodic (by either Definition~\ref{defn: t-periodicFock} or ~\ref{defn: t-periodic_partition}) then the following holds:
\begin{equation}\label{eq: S_matrix_periodic}
  \prod\limits_l\mathrm{S}_{l,l+t} = 1 \enspace .
\end{equation}

We also note an equivalent way to express Eq.~\ref{eq: transition_probability} that we use for some of our later results:

\begin{equation}\label{eq: probability_perm}
    P_{\vec{r} \rightarrow \vec{s}} := \sum\limits_{\rho \in \mathsf{S}_n} \left[ \prod\limits_{j=1}^n \mathrm{S}_{j, \rho(j)}\right] \mathsf{Per} \left(M \odot M_{\rho, \mathbb{I}}^*\right)  \enspace ,
\end{equation}
where $\odot$ denotes the element-wise Hadamard-product and $\mathsf{Per}(.)$ denotes the permanent matrix function, defined as:

\begin{equation}
    \mathsf{Per}(A) = \sum_{\sigma \in \mathsf{S}_n} \prod_i A_{i , \sigma(i)}
\end{equation}

For potentially multiply occupied input and output states, when the input is fully indistinguishable, the matrix elements $\mathrm{S}_{l,k} = 1$ and we have to account for the occupancy with a normalization factor $\frac{1}{\prod_i r_i! s_i!} $. The probability becomes~\cite{aaronsonComputational2011, tichySampling2015}:
\begin{equation}\label{eq: Ind_prob}
\begin{split}
    P_{\text{ind}} &= \frac{1}{\prod_i r_i! s_i!} \left(\sum\limits_{\sigma \in \mathsf{S}_n} \prod\limits_{j=1}^n M_{\sigma(j), j}\right)\left(\sum\limits_{\rho \in \mathsf{S}_n} \prod\limits_{j=1}^n M_{\rho(j), j}^*\right)
    \\ &= \frac{1}{\prod_i r_i! s_i!} \left|\sum\limits_{\sigma \in \mathsf{S}_n} \prod\limits_{j=1}^n M_{\sigma(j), j}\right|^2 \\
    &=\frac{1}{\prod_i r_i! s_i!} |\mathsf{Per}(M)|^2 \enspace .
\end{split}
\end{equation}
Using the same techniques, when the input is fully distinguishable, $\mathrm{S}_{l,k} = \delta_{l,k}$ and the probability becomes:
\begin{equation}\label{eq: dist_prob}
\begin{split}
    P_{\text{dist}} &= \sum\limits_{\sigma \in \mathsf{S}_n} \prod\limits_{j=1}^n \left(M_{\sigma(j), j} M_{\sigma(j), j}^* \right)\\ 
    &= \mathsf{Per}(|M^2|) \enspace .
\end{split}
\end{equation}
We will see later how to derive the transition probability for partition states, in Sec.~\ref{sec: transition_prob_derivation}.

\subsection{Zero-transmission laws}

Generalising the Hong-Ou Mandel effect for $n$ photons, Tichy et al.~\cite{tichyManyparticle2012, tichyZeroTransmission2010} determined general suppression laws for when a fully indistinguishable input state is passed through a multi-port balanced beamsplitter, the operator of which takes the form of the Quantum Fourier Transform (QFT). The corresponding $m \times m$ unitary matrix associated with this transformation can be defined as follows:
\begin{equation}\label{eq: QFt_{x(n)}ef}
    QFT_m = \frac{1}{\sqrt{m}}\begin{bmatrix}
        1 & 1 & \cdots & 1 \\
        1 & \omega^1 & \cdots & \omega^{m-1}\\
        1 & \omega^2 & \cdots & \omega^{2(m-1)} \\
        \vdots & \vdots & \ddots & \vdots \\
        1 & \omega^{m-1} & \cdots & \omega^{(m-1)(m-1)}
    \end{bmatrix} \enspace ,
\end{equation}

where $\omega=e^{2 \pi i/m}$. 

These general suppression rules are known as the zero-transmission laws (ZTL) and defined as follows.
\begin{itemize}
    \item [] \pZTL{} (\emph{periodic ZTL}): for any $t$-periodic input state $|r_t\rangle$ (see Definitions.~\ref{defn: t-periodicFock} and ~\ref{defn: t-periodic_partition}):
\begin{equation}\label{eq: pZTL}
    Q(\vec{s}) \neq 0 \ \text{mod} \ \left(\frac{m}{t}\right) \implies \langle s| \hat{QFT}|r_t\rangle = 0 \enspace , 
\end{equation}
where $\hat{QFT}$ represents the action of $QFT_m$ interferometer in Fock space \cite{aaronsonComputational2011}. 
This means events characterised by $\vec{s}$ are strictly suppressed. Then, by taking $t=1$, this reduces to:
    \item []\ZTL{}: for an $n=m$ fully indistinguishable input state with one photon per mode:
\begin{equation}\label{eq: ZTL}
    Q(\vec{s}) \neq 0  \implies \langle s|\hat{QFT}| r\rangle = 0 \enspace . 
\end{equation}
\end{itemize}

\subsection{Shift invariance property of the QFT for linear optics}\label{sec: shift_invariance}

The \textit{shift invariance} property of the QFT is well known in the qubit picture \cite{jozsaQuantum1998, chuangQuantum2010}. Moreover, it serves as a powerful tool for exploiting symmetries and will be used throughout this work. It is not immediately clear, however, whether this property also holds for a linear-optical Fourier interferometer. We therefore present a formal derivation to ensure its validity in the linear-optical setting.

Given an $n$-photon, $m$-mode input state $\vec{r}$, subjected to the $QFT_m$, the transition probability of ending up in state $\vec{s}$ 
for fully indistinguishable photons (Eq.\ref{eq: Ind_prob}) with $M = QFT_{d(\vec{r}),d(\vec{s})}$ is:

\begin{equation}
\label{eq:prob_perm}
    P_{\vec{r}\rightarrow \vec{s}}=  \frac{1}{\prod_j r_{j}!·s_{j}!} \left|\mathsf{Per}(QFT_{d(\vec{r}),d(\vec{s})})\right|^2=  \frac{1}{\prod_{j}^m r_{j}!·s_{j}!} \left|\sum_{\sigma\in S_n }\prod_{i=1}^n \frac{1}{\sqrt{m}}w^{d_{i}(\vec{s})·d_{\sigma(i)}(\vec{r})}\right|^2 \enspace ,
\end{equation}
where we have used the definition of the linear optical QFT in Eq.\ref{eq: QFT_def} to substitute for the elements of $QFT_{d(\vec{r}),d(\vec{s})}$.

Next, let $s^{(k)}_{i} = s_{i+k \text{ mod m}}$ be a shifted output of $\vec{s}$ that shifts each photon by $k$ modes, modulo $m$. Then, according to Definition~\ref{defn: mode_assignment} of the mode assignment vector, the elements of $d(\vec{s}^{\ (k)})$ can be expressed in terms of elements of $d(\vec{s})$ as:

\begin{equation}
    d_{\sigma(i)}(\vec{s}^{\ (k)}) = d_i(\vec{s})+k \text{ mod m} \enspace .
\end{equation}

We can now compute the transition probability from $\vec{r}$ to $\vec{s}^{\ (k)}$ as:

\begin{equation}\label{eq: prob_shifted}
    P_{\vec{r}\rightarrow \vec{s}^{(k)}}=  \frac{1}{\prod_j r_{j}!·s^{(k)}_{j}!} \left|\sum_{\sigma\in S_n }\prod_{i=1}^n \frac{1}{\sqrt{m}}\omega^{d_{i}(\vec{s}^{\ (k)})·d_{\sigma(i)}(\vec{r})}\right|^2 =  \frac{1}{\prod_j r_{j}!·s_{j}!} \left|\sum_{\sigma\in S_n }\frac{1}{\sqrt{m}}\omega^{\sum_{i=1}^n (d_{i}(\vec{s})+k)·d_{\sigma(i)}(\vec{r})}\right|^2 \enspace ,
\end{equation}

where we have used the fact that $\prod_j s_j! = \prod_j s_{j+k \text{ mod m}}!$. Using the Q-value, Definition~\ref{defn: q_value}, we can expand the sum in the exponent as:

\begin{equation}
    \sum_{i=1}^n (d_{i}(\vec{s})+k)·d_{\sigma(i)}(\vec{r}) = \sum_{i=1}^n d_{i}(\vec{s})·d_{\sigma(i)}(\vec{r})+\sum_{i=1}^nk·d_{\sigma(i)}(\vec{r}) =  \sum_{i=1}^n d_{i}(\vec{s})·d_{\sigma(i)}(\vec{r}) + k\cdot Q(\vec{r}) \enspace ,
\end{equation}

and then we substitute this back in to Eq.\ref{eq: prob_shifted}:

\begin{equation}\label{eq: prob_shifted_final}
    P_{\vec{r}\rightarrow \vec{s}^{(k)}}=  \frac{1}{\prod_j r_{j}!·s_{j}!} |\sum_{\sigma\in S_n }\frac{1}{\sqrt{m}}\omega^{\sum_{i=1}^n d_{i}(\vec{s})·d_{\sigma(i)}(\vec{r}) + k\cdot Q(\vec{r})}|^2 = \frac{1}{\prod_j r_{j}!·s_{j}!} |\omega^{k\cdot Q(\vec{r})}\sum_{\sigma\in S_n }\prod_{i=1}^n \frac{1}{\sqrt{m}}\omega^{d_{i}(\vec{s})·d_{\sigma(i)} (\vec{r})}|^2
\end{equation}

In  other words, if we shift the elements of $\vec{s}$ by $k$ then we add a global phase of $\omega^{k\cdot Q(\vec{r})}$ to the permanent, which does not impact the transition probability, as $|\omega^{k\cdot Q(\vec{r})}|^2=1$ and therefore:

\begin{equation}
     P_{\vec{r}\rightarrow \vec{s}^{\ (k)}} = P_{\vec{r}\rightarrow \vec{s}}
\end{equation}

\section{Main theoretical results}\label{sec: main_theory_results}

\subsection{Deriving the transition probability for partition states}\label{sec: transition_prob_derivation}
Our results rely on the fact that the transition probability (Eq.~\ref{eq: transition_probability}) for partition input states can be written as a sum of products of probabilities over the distinguishable subsets of photons. In this section, we prove that this is indeed the case. 

We saw earlier in Definition~\ref{defn: partition_state} that we can write our input partition state as: 
\begin{equation}\label{eq: input_state}
    |\phi(\vec{r})\rangle := \left( \prod\limits_{j=1}^K \prod\limits_{i=1}^m \frac{1}{\sqrt{r_i^{(j)}!}} (\hat{a}_{i, |\psi_j\rangle}^\dagger)^{(r_i^{(j)})}\right)|0\rangle \enspace .
\end{equation}
We can write any fixed output state $\vec{s}$ containing subsets of indistinguishable photons in fixed output modes, in the same form as our partition state (if the positions of the photons are known) in Definition~\ref{defn: partition_state} as $\vec{s} = \sum\limits_{j=1}^K \vec{s}^{\ (j)}$ where each $\vec{s}^{\ (j)} = (s_1^{(j)}, ..., s_m^{(j)})$. Using this formalism we define the following set for each output vector $\vec{s}$ that contains all the possible configurations of $\vec{s}^{\ (j)}$:
\begin{equation}
    \mathcal{S}(\vec{s}) := \left\{\{\vec{s}^{\ (j)}\} | \sum_j \vec{s}^{\ (j)} = \vec{s}\right\} \enspace ,
\end{equation}
where, of course and by definition of a lossless system, the number of photons in each output register is the same as that in each input register, i.e. $\sum\limits_i s_i^{(j)} = \sum\limits_i r_i^{(j)} = n_j \ \forall j$. We now present our result for the transition probability of partition states.

\begin{proposition}\label{prop: probability}
    For partition input states where we have $K$ subsets of fully indistinguishable photons, the transition probability is:
    \begin{equation}\label{eq: transition_probability}
    P_{\vec{r} \to \vec{s}}= \sum\limits_{\{\vec{s}^{(j)}\} \in \mathcal{S}(\vec{s})}  \left(\prod_{j= 1}^K N_j|\mathsf{Per}(M_{d(\vec{r}^{(j)}), d(\vec{s}^{(j)})})|^2\right) \enspace ,
    \end{equation}
where the normalization factor for each register is $N_j = \frac{1}{\prod_i r_i^{(j)} ! s_i^{(j)}!}$.
\end{proposition}
\begin{proof}
    We know that, since our system is lossless, the output state vector will also be a partition state over the registers, $\vec{s} = \sum\limits_{j=1}^K \vec{s}^{\ (j)}$.
    
    If the exact configuration of photons in the output is known, i.e. the specific set 
    $\{\vec{s}^{(j)}\}$, then using Eq.~\ref{eq: input_state}, we can write the transition probability from $\vec{r}$ to $\vec{s}_{\{\vec{s}^{(j)}\}}$ (Eq.~\ref{eq: transition_prob}) as:
    \begin{equation}
        P_{\vec{r} \rightarrow \vec{s}_{\{\vec{s}^{(j)}\}}} = \left|\langle 0| \prod\limits_{j=1}^K \prod\limits_{k=1}^m \frac{1}{\sqrt{s_k^{(j)}!}} (\hat{a}_{k, |\psi_j\rangle})^{s_k^{(j)}} U \left(\prod\limits_{j=1}^K \prod\limits_{i=1}^m\frac{1}{\sqrt{r_i^{(j)}!}}(\hat{a}_{i, |\psi_j\rangle}^\dagger)^{r_i^{(j)}}|0\rangle\right)\right|^2  \enspace ,
    \end{equation}
where for ease of notation, we write the unitary as $U(\cdot)$ which acts on the creation and annihilation operators as specified in Eq.~\ref{eq: unitary_transform}. The probability then simplifies to:
\begin{equation}
    \begin{split}
        P_{\vec{r} \rightarrow \vec{s}_{\{\vec{s}^{(j)}\}}} &= \prod\limits_{j=1}^K  \frac{1}{\prod\limits_{i=1}^m s_i^{(j)}!r_i^{(j)}!}\left|\langle 0|\prod\limits_{l=1}^{n_j}\left(\hat{a}_{d_l\left(\vec{s}^{(j)}\right)}\right)U \left(\hat{a}_{d_l(\vec{r}^{(j)})}^\dagger\right)|0\rangle\right|^2 \\
        &= \prod\limits_{j=1}^K  \frac{1}{\prod\limits_{i=1}^m s_i^{(j)}!r_i^{(j)}!}\left|\langle \vec{s}^{\ (j)}| U|\vec{r}^{\ (j)}\rangle\right|^2
    \end{split}
\end{equation}

From Eq.~\ref{eq: transition_prob} and Eq.~\ref{eq: Ind_prob} it should be clear to see that the term after the product $\prod\limits_{j=1}^K$, i.e. the normalization with the modulus-squared, is exactly equal to the transition probability from state $\vec{r}^{(j)}$ to state $\vec{s}^{(j)}$ that contain the same fully indistinguishable photons, which according to Eq.~\ref{eq: Ind_prob} is:
\begin{equation}
       \frac{1}{\prod\limits_{i=1}^m s_i^{(j)}!r_i^{(j)}!}\left|\langle \vec{s}^{\ (j)}| U|\vec{r}^{\ (j)}\rangle\right|^2 = N_j |\mathsf{Per}(M_{d(\vec{r}^j), d(\vec{s}^j)})|^2 \enspace ,
\end{equation}
so that our probability to a known output configuration state $\vec{s}_{\{\vec{s}^{(j)}\}}$ is just:
\begin{equation}
    P_{\vec{r} \rightarrow \vec{s}_{\{\vec{s}^{(j)}\}}} = \prod\limits_{j=1}^K N_j |\mathsf{Per}(M_{d(\vec{r}^j), d(\vec{s}^j)})|^2 \enspace ,
\end{equation}
then by noting that we must sum over all possible configurations $ \{\vec{s}^{(j)}\} \in \mathcal{S}(\vec{s})$ to get the transition probability of fixed input $\vec{r}$ to output $\vec{s}$ we complete the proof: 
\begin{equation}
    P_{\vec{r} \to \vec{s}}= \sum\limits_{\{\vec{s}^{(j)}\} \in \mathcal{S}(\vec{s})}  \left(\prod_{j= 1}^K N_j|\mathsf{Per}(M_{d(\vec{r}^{(j)}), d(\vec{s}^{(j)})})|^2\right) \enspace .
\end{equation}
\end{proof}
The above derivation results straight-forwardly from Definition~\ref{defn: partition_state} of a partition state, and then using the formalisms of Aaronson and Tichy~\cite{aaronsonComputational2011, tichySampling2015}.

To convince ourselves that this is indeed the correct transition probability, we can also use the above to derive the coincidence transition probability of Tichy, Eq.~\ref{eq: transition_prob}. If we have only single photons in each mode, in the input and output then Eq.~\ref{eq: transition_probability} loses the normalization factor, and using Eq.~\ref{eq: Ind_prob} we have: 
\begin{equation}
\begin{split}\label{eq: deriving_tichy}
    P_{\vec{r} \rightarrow \vec{s}} &= \sum\limits_{\{\vec{s}^{(j)}\} \in \mathcal{S}(\vec{s})} \left(\prod\limits_{j=1}^K |\mathsf{Per}(M_{d(\vec{r}^{(j)}), d(\vec{s}^{(j)})})|^2\right) \\ 
    &=\sum\limits_{\{\vec{s}^{(j)}\} \in \mathcal{S}(\vec{s})} \left(\prod\limits_{j=1}^K \left|\sum_{\sigma_j \in \mathsf{S}_{n_j}} \prod\limits_{i = 1}^{n_j} M_{\sigma_j(i), i}\right|^2\right) \\ 
    &= \sum\limits_{\{\vec{s}^{(j)}\} \in \mathcal{S}(\vec{s})}\sum_{\sigma_j,\rho_j \in \mathsf{S}_{n_j}} \prod\limits_{j=1}^K  \prod\limits_{i = 1}^{n_j} M_{\sigma_j(i), i} M_{\rho_j(i), i}^* \enspace .
\end{split}
\end{equation}
In the above equation the product of the  distinguishability matrix elements are equal to one: $\prod\limits_{i=1}^{n_j} \mathrm{S}_{\rho_j(i), \sigma_j(i)} = 1$ for each register of indistinguishable photons. However, we could easily add this back into our equation and then, since any orthogonal photons result in a zero-element of the distinguishability matrix, we can see that the product of distinguishability matrix elements in Tichy's Eq.~\ref{eq: transition_prob} can be split into the following product:
\begin{equation}
    \prod\limits_{k=1}^n \mathrm{S}_{\rho(k), \sigma(k)}  = \prod\limits_{j=1}^K \prod\limits_{i=1}^{n_j} \mathrm{S}_{\rho_j(i), \sigma_j(i)} \enspace ,
\end{equation}
and applying the same treatment to our middle term for the effective scattering matrices, we can replace it with:
\begin{equation}
    \prod\limits_{j=1}^K\prod\limits_{i=1}^{n_j} M_{\sigma_j(i), i} M_{\rho_j(i), i}^*= \prod\limits_{k=1}^n M_{\sigma(k), k} M_{\rho(k), k}^*
    \enspace .
\end{equation}
Then we can focus on the sum over the possible sets of distinguishable registers and over the permutations of those specific registers. The sum in Eq.~\ref{eq: transition_prob} is over all possible permutations of $n$ photons, the permutation group $\mathsf{S}_n$, where the other terms in the sum account for the distinguishabilities of the photons. In our Eq.~\ref{eq: deriving_tichy} we first sum over the permutation group $\mathsf{S}_{n_j}$ of $n_j$ photons, the number of photons in each $j_{th}$ register and then sum over all the sets of register-states $\{\vec{s}^{(j)}\}$ that give our fixed output $\vec{s}$. This is equivalent to summing over all possible permutations of the total $n$ photons as it will cover all possible paths for the single photons to end up in each output mode. Therefore, we have that:
\begin{equation}
\begin{split}
    \sum\limits_{\sigma, \rho \in \mathsf{S}_n}(.) &= \sum\limits_{\{\vec{s}^{(j)}\} \in \mathcal{S}(\vec{s})}\sum_{\sigma_j,\rho_j \in \mathsf{S}_{n_j}} (.),\\ then 
    \sum\limits_{\sigma, \rho \in \mathsf{S}_n} \prod\limits_{j=1}^n M_{\sigma(j), j} M_{\rho(j), j}^* \mathrm{S}_{\rho(j), \sigma(j)} &=\sum\limits_{\{\vec{s}^{(j)}\} \in \mathcal{S}(\vec{s})}\sum_{\sigma_j,\rho_j \in \mathsf{S}_{n_j}} \prod\limits_{j=1}^K  \prod\limits_{i = 1}^{n_j} M_{\sigma_j(i), i} M_{\rho_j(i), i}^* \mathrm{S}_{\rho_j(i), \sigma_j(i)} \enspace .
\end{split}
\end{equation}

Annoni and Wein~\cite{annoniIncoherent2025} show a similar transition probability for partition states, with separability based on permutation symmetries.

\subsubsection{Necessary probabilities}\label{sec: probabilities}
The probability that we are interested in is that of the Q-value (Definition~\ref{defn: q_value}) which is a function of a specific output $Q(\vec{s})$. We can further express the Q-value as a function of the different register-states. First, we define a quantity that is the sum of the elements of the mode-assignment vector (Definition~\ref{defn: mode_assignment}) for each register-state $\vec{s}^{(j)}$ where $j=1, ..., K$:
\begin{equation}
    D_j := \sum\limits_{l=1}^{n} d_l(\vec{s}^{\ (j)}) \enspace ,
\end{equation}
then we can write the Q-value for each register as:
\begin{equation}
    Q_j := Q(\vec{s}^{\ (j)}) = \text{mod} \left( D_j, m \right) \enspace .
\end{equation}
We then have a clear way of expressing the Q-value for a specific output in terms of the registers and their Q-values:
\begin{equation}
\begin{split}\label{eq: Q_sub_mods}
    Q(\vec{s}) &= \mod \left(\sum\limits_{j=1}^K D_j, m\right) \\
    &= \mod \left(\sum\limits_{j=1}^K (\text{mod} \ (D_j, m)), m \right)\\ 
    &= \mod \left(\sum_{j=1}^K Q_j, m\right)\enspace .
\end{split}
\end{equation}

Now, we can define the probability of obtaining a specific register-sum $Q_j$ value from all the possible output vectors associated with that register $\vec{s}^{\ (j)}$, for a fixed input register-state $\vec{r}^{\ (j)}$.
\begin{equation}\label{eq: Q_value_prob}
    P_{\vec{r}^{(j)}}(Q_j = x) = \sum\limits_{\vec{s}^{\ (j)} |Q_j = x} P_{\vec{r}^{(j)} \rightarrow \vec{s}^{(j)}} \enspace .
\end{equation}

Similarly, we can define the probability of obtaining a specific $Q$-value (from all possible output vectors) for a fixed input partition state $\vec{r}$ as:
\begin{equation}\label{eq: prob_Q}
    P_{\vec{r}}(Q=y) := \sum\limits_{\{\vec{s}\}|Q(\vec{s}) = y} P_{\vec{r} \rightarrow \vec{s}} \enspace ,
\end{equation}
where $\{\vec{s}\}|Q(\vec{s}) = y$ is the set of all output state-vectors $\vec{s}$ that result in $Q(\vec{s}) = y$.

Finally, and equivalently we can define the probability of obtaining a specific Q-value with respect to the registers, as:
\begin{equation}\label{eq: prob_Q_value}
    P_{\vec{r}}(Q=y) := \sum\limits_{\{Q_j\}|Q=y} \prod\limits_{j=1}^K P_{\vec{r}^{(j)}} (Q_j) \enspace ,
\end{equation}
where $\{Q_j\}|Q(\vec{s}) = y$ is the set of all $Q_j$ such that $Q(\vec{s}) = y$, by Eq.~\ref{eq: Q_sub_mods}.

\subsubsection{Q-marginals probability for partition states}
An important result that we use for the main results of Sec.~\ref{sec: partition_results} is what we refer to as the $Q$-marginals probability:
\begin{theorem}[Q-marginals probability]
\label{th:q_marginal}
    Given an input partition state $\rho$, and a QFT interferometer with $n=m$ photons and one photon per mode, the probability of obtaining a specific $Q$-value (from all possible output vectors) for a fixed input partition state $\vec{r}$ is:
    \begin{equation}
    \label{eq:q_marginal}
        P_{\vec{r}}(Q=y) = \frac{1}{n}\sum_{j=0}^{n-1}\omega^{yj} \left[ \prod\limits_{i=1}^n \mathrm{S}_{i, i+j}\right]
    \end{equation}
\end{theorem}
\begin{proof}
    Using Eq.~\ref{eq: probability_perm}, we can re-write the probability of obtaining a specific $Q$-value (see Eq.~\ref{eq: prob_Q}) as:
    \begin{equation}\label{eq: prob_Q_y}
        P_{\vec{r}}(Q=y) = \sum\limits_{\{\vec{s}\}|Q(\vec{s}) = y}\frac{1}{\prod_j s_j!}\sum\limits_{\tau \in \mathsf{S}_n} \left[\prod\limits_{j=1}^n \mathrm{S}_{j, \tau(j)}\right] \mathsf{Per}(M_{i, d_j(\vec{s})} \odot M_{\tau(i), d_j(\vec{s})}^*) \enspace ,
    \end{equation}
    where we have used the fact that for $n=m$ photons and a single photon in each input mode, i.e. $d_j(\vec{r}) \in (1, ..., n)$ then we write the effective scattering matrix as: $M_{j, d_j(\vec{s})} \in \mathbb{C}^{n \times n}$ since $d_j(\vec{r}) = (1, ..., n)$ always.
    
    For clarity we denote the term in Eq.~\ref{eq: prob_Q_y} that multiplies $\prod_jS_{j,\tau(j)}$ as $c_\tau(\vec{s}) = \mathsf{Per}(M_{i, d_j(\vec{s})} \odot M_{\tau(i), d_j(\vec{s})}^*)$ and define an equivalent coefficient for the $Q$-value:
    \begin{equation}
        c_\tau(Q=y) := \sum\limits_{\{\vec{s}\}|Q(\vec{s}) = y} \frac{1}{\prod_j s_j!} \mathsf{Per}(M_{i, d_j(\vec{s})} \odot M_{\tau(i), d_j(\vec{s})}^*) 
    \end{equation}
    
    Substituting the effective QFT scattering matrix $M_{i, d_j(\vec{s})} := \frac{1}{\sqrt{n}} \omega^{i d_j(\vec{s})} \enspace$, where, of course, $\omega = e^{2\pi i/n}$.
    
    We  can expand the permanent term as:
    \begin{equation}
    \begin{split}\label{eq: simplify_perm}
        \mathsf{Per}(M_{i, d_j(\vec{s})} \odot M_{\tau(i), d_j(\vec{s})}^*)   
        &=   \sum_{\sigma \in \mathsf{S}_n}\prod_{i=1}^{n}\frac{1}{n}\omega^{(i-\tau(i))d_{\sigma(i)}(\vec{s})} \\
        &= \frac{1}{n^n} \sum_{\sigma \in \mathsf{S}_n} \omega^{\sum \limits_i (i-\tau(i))d_{\sigma(i)}(\vec{s})}\enspace ,
    \end{split}
    \end{equation}
    where we have collected the constant prefactor $\frac{1}{n^n}$ and rewritten the product of roots of unity as a single term.
    
    In order to get to our result we introduce the following Lemma.
    \begin{lemma}\label{lem: c_tau}
        The value of the coefficient $c_\tau(Q=y)$ depends only on shift permutations:
        \begin{equation}
        \label{eq:c_tau_0}
            c_\tau(Q=y) = \begin{cases}
                \frac{\omega^{yj}}{n}  & \text{if } \tau(i) = i+j \mod n \\
                0 & \text{otherwise}
            \end{cases}
        \end{equation}
    \end{lemma}
    To prove this, we re-write our coefficient using Eq.~\ref{eq: simplify_perm}, defining the permutation vector $p_i(\tau) = i-\tau(i)$ for ease of notation:
    \begin{equation}
    \label{eq:c_tau}
        c_\tau(Q=y) = \frac{1}{n^n} \sum\limits_{\{\vec{s}\}|Q(\vec{s})=y}\frac{1}{\prod\limits_j s_j!}\sum\limits_{\sigma \in \mathsf{S}_n}\omega ^{\sum\limits_i p_i d_{\sigma(i)}(\vec{s})} \enspace .
    \end{equation}

    For a fixed outcome $\vec{s}$ for any $s_j > 1$ there is a corresponding entry in $d(\vec{s})$ that is repeated $s_j > 1$ times. Therefore, $d(\vec{s})$ will be symmetric under $s_j!$ permutations that shuffle these repeated elements. Extending this argument to $s_j>0$ and noting that permutations that move distinct groups of elements necessarily commute, we find that the subgroup of the symmetries $\mathcal{S}_{\vec{s}} \preceq \mathsf{S}_n$ of $d(\vec{s})$ consists of $\prod_j s_j!$ different permutations.

    Factoring out these symmetries, only $n!/\prod_j s_j!$ distinct permutations of $d(\vec{s})$ are possible and they are listed by the elements of the factor group $\mathsf{S}_n/\mathcal{S}_{\vec{s}}$. This construction means we can enumerate the $n!$ permutations as products $\sigma \sigma'$, with $\sigma\in \mathsf{S}_n/\mathcal{S}_{\vec{s}},\sigma'\in \mathcal{S}_{\vec{s}}$. Since there are $\sigma'$ symmetries of $d(\vec{s})$, we are able to extract a factor $|\mathcal{S}_{\vec{s}}| =\prod_j s_j!$ from the summation in Eq.~\ref{eq:c_tau} :
    \begin{equation}
        \sum_{\sigma \in \mathsf{S}_n} (\hdots) =  \sum_{\sigma \in \mathsf{S}_n/\mathcal{S}_{\vec{s}}}\sum_{\sigma' \in \mathcal{S}_{\vec{s}}} (\hdots) = \prod_j s_j \sum_{\sigma \in \mathsf{S}_n/\mathcal{S_{\sigma}}}(\hdots)
    \end{equation}

    Resulting in the following rewriting:
    \begin{equation}
    \label{eq:c_tau_2}
        c_\tau(Q=y) = \frac{1}{n^n} \sum\limits_{\{\vec{s}\}|Q(\vec{s})=y}\sum\limits_{\sigma \in \mathsf{S}_n/\mathcal{S}_{\vec{s}}}\omega ^{\sum\limits_i p_i d_{\sigma(i)}(\vec{s})} \enspace .
    \end{equation}
    
    The double sum $\sum\limits_{\{\vec{s}\}|Q(\vec{s})=y} \sum\limits_{\sigma \in \mathsf{S}_n/\mathcal{S}_{\vec{s}}}$ enumerates over the distinct permutations of all distinct outcome vectors $d(\vec{s})$ that produce a $Q=y$.
    
    The pairs $(d(\vec{s}),\sigma)$ are in bijection with the the tuples $(d_1, d_2, \hdots, d_{n-1}) \in \{0,\hdots,n-1\}^{n-1}$ where $d_i$ (namely the position of photon $i$) is freely chosen in $\{0,\hdots,n-1\}$ and $d_n$ is fixed to the unique value that produces $Q(\vec{s}) = y$. The application of $\sigma$ to $d(\vec{s})$ produces a unique element $\vec{d} \in \{0,\hdots,n-1\}^{n-1}$, while a pair $(d(\vec{s}), \sigma)$ can be constructed choosing $\sigma^{-1}$ to be the unique permutation (up to symmetries) that sorts $\vec{d}$, so that $d(\vec{s}) = \sigma^{-1} \circ \vec{d}$.
    
    The condition that $Q(\vec{s}) = y$ can then be enforced by:
    \begin{equation}
        d_n = y - \sum_{i=1}^{n-1}d_i \mod n
    \end{equation}

    Substituting $d_{\sigma(i)}(\vec{s}) \rightarrow \vec{d}$ into Eq.~\ref{eq:c_tau_2}, we obtain:
    
    \begin{equation}
    \label{eq:c_tau_3}
    \begin{split}
        c_\tau(Q=y) &= \frac{1}{n^n} \sum_{d_1,\hdots,d_{n-1} \in \{0,\hdots,n-1\}}\omega ^{\sum\limits_{i=1}^n p_i d_i} \enspace \\
        & = \frac{\omega^{yp_n}}{n^n} \sum_{d_1,\hdots,d_{n-1} \in \{0,\hdots,n-1\}}\omega ^{\sum\limits_{i=1}^{n-1} (p_i-p_n) d_i} \enspace\\
        & = \frac{\omega^{yp_n}}{n^n} \sum_{d_1 \in \{0,\hdots,n-1\}}\omega ^{ (p_1-p_n) d_1} \sum_{d_2 \in \{0,\hdots,n-1\}}\hdots\sum_{d_{n-1} \in \{0,\hdots,n-1\}}\omega ^{ (p_{n-1}-p_n) d_{n-1}} \enspace\\
    \end{split}
    \end{equation}
    where in the first step we've operated the substitution $p_n = y-\sum_i d_i$ and in the second separated the summation over productories into a product of independent summations.
    Each of these summations over the roots of unity will be equal to 0 unless the following conditions on $p_i$ are met:
    \begin{equation}
        \sum_{d_i \in \{0,\hdots,n-1\}}\omega ^{ (p_i-p_n) d_i} = \begin{cases}
            n &\text{if } p_i =p_n\\
            0 & \text{otherwise}
        \end{cases}
    \end{equation}

    Leading to only $n$ allowed configurations for the permutation $\tau$:
    \begin{equation}
        p_i = j \text{ for } j \in \{0,\hdots,n-1\}
    \end{equation}

    Which translates to:

    \begin{equation}
    \begin{split}
        c_\tau (Q=y) = \frac{\omega^{yj}}{n} \quad \text{ for } p_i = j
    \end{split}
    \end{equation}

    The solution $p_i = j$ corresponds to the shift permutation $\tau(i) = i+j \mod n$. Now, we can substitute the value of $c_\tau(Q=y)$ into Eq.~\ref{eq: prob_Q_y}, taking the sum over $j=0, ..., n-1$ and noting that the distinguishability matrix term becomes $\prod\limits_{i=1}^n \mathrm{S}_{i, i+j}$ since the only permutation that contributes is $\tau(i) = i + j \mod n$. Therefore, we obtain the claim of Lemma~\ref{lem: c_tau}.

\end{proof}

\subsection{Quantifying the genuine $n$-photon indistinguishability ($c_1$) when $n=m$ with one photon per mode}\label{sec: c_1 for n=m}

In this section, we prove our main results. Let $\rho$ be an unknown partially distinguishable mixed state as seen in Eq.~\ref{eq: partially_distinguishable} with $n=m$ photons and exactly one photon per mode, repeated here for convenience: 
\begin{equation}\label{eq: density_matrix}
    \rho = c_1 \rho^\parallel + \sum\limits_k c_k \rho_k^\perp \enspace ,
\end{equation}

Focusing on the second term $\sum_k c_k \rho_k^\perp$, if we assume that each $\rho_k^\perp$ has an equal probability given a specific post-selection strategy then we can reformulate Eq.~\ref{eq: density_matrix}. We do this by introducing the fact that we post-select on states $\vec{s}$ where $Q(\vec{s}) \neq 0$ and therefore the probability of retrieving such an output state is: 

\begin{equation}\label{eq: c_1_probs}
    P(Q\neq0)=c_1 P_{i}(Q\neq0) + (1-c_1)P_{d}(Q\neq0)
\end{equation}

where $P_d$ denotes the probability of any partially distinguishable partition state, and $P_i$ the probability of the fully indistinguishable state.
 
Due to the \ZTL{} (see Eq.~\ref{eq: ZTL}), we know that in this setting $P_i(Q\neq0) = 0$ and we prove in this section that when $\rho_k^\perp$ are OBB states (for all $n$), or if they are general partition states (when $n$ is prime) then $P_d(Q\neq0) = 1-1/n$ and therefore

\begin{equation}
\begin{split}
P(Q\neq0) &= (1-c_1)(1-1/n) \\
c_1 &= \frac{(1-1/n)-P(Q\neq0)}{1-1/n} \enspace ,
\end{split}
\end{equation}

which means we can compute $c_1$ by determining $P(Q\neq0)$ experimentally and substituting in our analytical probabilities. 

\subsubsection{For OBB input states}\label{sec: OBB}
\begin{theorem}[OBB probability]\label{thm: OBB}
    Given an initial $n$-photon state $\rho$ of the form of Eq.~\ref{eq: density_matrix} subjected to an interferometer that implements an $n\times n$ QFT matrix, the total transition probability of observing output events $\vec{s}$ such that $Q(\vec{s}) \neq 0$ (see Definition~\ref{defn: q_value}) for $\rho_k^{\perp}$ is $P(Q\neq 0) = 1 - 1/n$ for all $k$ and for all $n$, if they are OBB partition states.
\end{theorem}
\begin{proof}
    In order to show that $P(Q\neq 0) = 1 - 1/n$, we first determine the probability $P(Q=0)$. Partition states are made up of $K$ subsets of distinguishable or indistinguishable photons that we refer to as registers. An OBB partition state, as seen in Definition~\ref{defn: obb_state} consists of a single register of fully indistinguishable photons and then $K-1$ registers containing exactly one distinguishable photon that is fully orthogonal to all others. 
    
    Using Eq.~\ref{eq: prob_Q_value} we see that the total probability is:
    \begin{equation}
        P(Q=0) = \sum\limits_{\{Q_j\}|Q=0} \prod\limits_{j=1}^{K} P_{\vec{r}^{(j)}}(Q_j) \enspace ,
    \end{equation}

    now, since we know from Eq.~\ref{eq: Q_sub_mods} that $Q(\vec{s}) = \mod \left( \sum\limits_{j=1}^K Q_j, m\right)$ and knowing that $Q(\vec{s}) = 0, ..., n-1$, we can then split the above equation as follows:
    \begin{equation}
    \begin{split}
        P(Q=0) &= P_{\vec{r}^{(1)}}(Q_1=0)\sum\limits_{\{Q_j\}|\sum_j Q_j=0} \prod\limits_{j=1}^{K-1} P_{\vec{r}^{(j)}}(Q_j) \\
        &+ P_{\vec{r}^{(1)}}(Q_1=1)\sum\limits_{\{Q_j\}|\sum_j Q_j=n-1} \prod\limits_{j=1}^{K-1} P_{\vec{r}^{(j)}}(Q_j) \\
        &+\cdots + P_{\vec{r}^{(1)}}(Q_1=n-1)\sum\limits_{\{Q_j\}|\sum_j Q_j=1} \prod\limits_{j=1}^{K-1} P_{\vec{r}^{(j)}}(Q_j)
    \end{split} \enspace ,
    \end{equation}
    where register $1$ and register-state $\vec{r}^{(1)}$ contains the indistinguishable photons. For each register containing one distinguishable photon, after applying the QFT there is an equal probability that the photon will end up in any of the output modes and therefore the probability of any Q-value for registers $j=1, ..., K-1$ is $1/n$ and $\prod\limits_{j=1}^{K-1} 1/n = 1/n^{K-1}$, which leads to: 
    \begin{equation}
    \begin{split}
        P(Q=0) &= P_{\vec{r}^{(1)}}(Q_1=0)\sum\limits_{\{Q_j\}|\sum_j Q_j=0} 1/n^{K-1} \\
        &+ P_{\vec{r}^{(1)}}(Q_1=1)\sum\limits_{\{Q_j\}|\sum_j Q_j=n-1} 1/n^{K-1} \\
        &+\cdots + P_{\vec{r}^{(1)}}(Q_1=n-1)\sum\limits_{\{Q_j\}|\sum_j Q_j=1}  1/n^{K-1}
    \end{split} \enspace .
    \end{equation}
    We are now interested in the set of Q-values $\{Q_j\}|\sum_j Q_j = 0, ..., n-1$. It is well-known that for any fixed sum $x \in \mathbb{Z}_n$, the number of tuples $(x_1, ..., x_l) \in \mathbb{Z}_n^l$ satisfying $x_1 + ... + x_l \equiv x \mod n$ is exactly $n^{l-1}$~\cite{taoAdditive2006}. Therefore, for our $K-1$ values of $Q_j$, there are exactly $n^{K-2}$ arrangements that give our desired outputs. Since each sum over distinguishable registers gives the same number, we can group the probabilities for the indistinguishable register together, leading to:
    \begin{equation}
    \begin{split}
        P(Q=0) &= (P_{\vec{r}^{(1)}}(Q_1=0) + P_{\vec{r}^{(1)}}(Q_1=1) + \cdots P_{\vec{r}^{(1)}}(Q_1=n-1)) \frac{n^{K-2}}{n^{K-1}} \\ 
        &= 1 \cdot \frac{n^{K-2}}{n^{K-1}} = 1/n
    \end{split} \enspace ,
    \end{equation}
    where for the last step we have used the fact that the term in the brackets is a convex sum. Thus we complete the proof:
    \begin{equation}
        P(Q\neq0) = 1 - P(Q=0) = 1 - 1/n 
    \end{equation}
\end{proof}

\subsubsection{Computing $c_1$ for all partition states and $n$ prime}\label{sec: partition_results}
In order to get to our main results for general partition states, we first introduce a Lemma for the distinguishability matrix of $t$-periodic states.

\begin{lemma}
\label{lm:period}
Let $|\phi(\vec{r}_t)\rangle$ be a $t$-periodic partition state with an associated distinguishability matrix $\mathrm{S}$ and $j\in \mathbb{N}$, then:

\begin{equation}
    \prod_i S_{i,i+j} = 1 \Leftrightarrow j = kt \text{ for } k \text{ an integer}
\end{equation}

\end{lemma}
\begin{proof}
    We know from Eq.~\ref{eq: S_matrix_periodic} that for $t$-periodic states $\prod\limits_i \mathrm{S}_{i, i+t} = 1$. Now, given that $\langle a|b \rangle = 1$ implies $\ket{a} = \ket{b}$ the following transitive property holds for the $\mathrm{S}$ matrix:
    \begin{equation}
        \mathrm{S}_{ab} =1, \mathrm{S}_{bc} = 1 \implies \mathrm{S}_{ac} = 1 \enspace .
    \end{equation}

    By repeated application of this property to the $\mathrm{S}$-matrix of a $t$-periodic state, we find:    
    \begin{equation}
    \label{eq: transitive}
    \begin{split}
        &\mathrm{S}_{i, i+t} = 1, \mathrm{S}_{i+t, i+2t} = 1 \implies \mathrm{S}_{i, i+2t} = 1 \\
        \text{ therefore }\  &\mathrm{S}_{i,i +kt} = 1 \text{ for } k\geq 0 
    \end{split}
    \end{equation}

    Finally, we prove the stated claim by contradiction. If we assume that there is a $t'\neq t$ such that $\prod_i S_{i,i+t'} = 1$ and $\mathrm{gcd}(t,t') < t$, applying the same argument of Eq.~\ref{eq: transitive} we deduce that $S_{i, i+kt +k't'} = 1$ for all $i,k,k'$.
    
    By applying the Euclidean algorithm~\cite{hardyIntroduction2008}, we can always find some integers $k,k'$ such that $kt + k't' = \mathrm{gcd}(t, t')$. It follows that $\mathrm{S}_{i, i+k \cdot\mathrm{gcd}(t,t')} = 1$ and we can infer that the state has a periodicity strictly smaller than $t$, leading to a contradiction.
\end{proof}

Joining the lemma above with Thm.~\ref{th:q_marginal} we obtain an important characterization of the behavior of the Q-marginals:

\begin{theorem}[Uniform probability for $t$-periodic partition states]\label{thm: uniform_t_probability}
Given an initial $t$-periodic partition state as in Definition~\ref{defn: t-periodic_partition}, for $n=m$ photons and one photon per mode, the transition probability for any $Q$-value is:
\begin{equation}
    P_{\vec{r}_t}(Q=y) = \begin{cases}
        \frac{1}{t} \ \text{if} \ k=0 \mod \frac{n}{t} \\ 
        0 \  \text{otherwise}
    \end{cases} \enspace ,
\end{equation}
\end{theorem}

\begin{proof}
    Rewriting Eq.~\ref{eq:q_marginal} with only the non-zero terms $j = kt, k\geq 0$ we have:
    \begin{equation}
    \begin{split}
         P_{\vec{r}_t}(Q=y) &= \frac{1}{n}\sum_{k=0}^{n/t}\omega^{tky} \\
        &= \frac{1}{n}\sum_{k=0}^{n/t}(\omega^{t})^{ky} \\
        & = \frac{1}{t} \delta_{y \text{ mod } \frac{n}{t}} \\
    \end{split}
    \end{equation}

Where the last line follows from summing the $n/t$-th roots of unity, $\delta_{y \text{ mod } \frac{n}{t}}=1$ if $y = 0 \text{ mod } \frac{n}{t}$, and 0 otherwise.
\end{proof}

Now, we are in a position to compute $c_1$ for all parition states. Analogously to Thm.~\ref{thm: OBB}, when we have $n$-prime, we must prove that all the possible partition inputs result in output states where $P(Q\neq0) = 1-1/n$. We introduce the following theorem:
\\

\begin{theorem}[Prime $n$ partitions probability]
    Given an initial $n$-photon state $\rho$ of the form of Eq.~\ref{eq: density_matrix} subjected to an interferometer that implements an $n\times n$ QFT matrix, the total transition probability of observing output events $\vec{s}$ such that $Q(\vec{s}) \neq 0$ (see Definition~\ref{defn: q_value}) for $\rho_k^{\perp}$ is $P(Q\neq 0) = 1 - 1/n$ for all $k$ and for $n$-prime.
\end{theorem}

\begin{proof}
    When $n$ is prime, any partition state is what we call \emph{non-periodic} which means it either has period $t=1$ or period $t=n$. This can be seen directly by looking at Definition~\ref{defn: t-periodic_partition} and noticing that if $n$ is prime, then there does not exist an integer $t$ such that $n/t \in \mathbb{Z}$, i.e. the period can only be $t=n$ if we have a partition state or $t=1$ for the fully indistinguishable state. 

    Using the same logic as Theorem~\ref{thm: OBB}, we will determine the value of $P(Q=0)$ first. In this case, we can use an equivalent expression to Eq.~\ref{eq: prob_Q_value}, which is Eq.~\ref{eq: prob_Q} that we re-write here for the case when $Q=0$
    \begin{equation}\label{eq: repeat_prob_Q}
        P_{\vec{r}}(Q=0) := \sum\limits_{\{\vec{s}\}|Q(\vec{s}) = 0} P_{\vec{r} \rightarrow \vec{s}} \enspace ,
    \end{equation}
    and now from Theorem~\ref{thm: uniform_t_probability} we know that all $t$-periodic partition states have the property that $P(Q=y) = 1/t$ if $y= 0 \ \text{mod} \ n/t$ and since $t=n$ this means that $P(Q=y) = 1/n$ since any integer modulo one is equal to zero and Eq.~\ref{eq: repeat_prob_Q} simply becomes:
    \begin{equation}
        P(Q=0) = 1/n \enspace ,
    \end{equation}
    which completes the proof.
\end{proof}

\subsubsection{Computing $c_1$ for all partition states and non-prime $n$}\label{sec: non-prime_partition_results}
According to Theorem~\ref{thm: uniform_t_probability} given an input state $\rho$ of the form of Eq.~\ref{eq: density_matrix} and applying an $n\times n \ QFT$, the output probability of any $Q$-value, $P_{\vec{r}_t}(Q=i)$, only depends on the period $t$ of each partition in the mixture. Since a state with $n$ a non-prime does not necessarily have period $t=n$, we can re-write Eq.\ref{eq: density_matrix} by decomposing it into partition states with different periodicities:

\begin{equation}\label{eq:2}
    \rho = c_{t_1}\cdot \rho_{t_1} + \sum_j c_{j,t_2}\cdot \rho_{j,t_2} + ...+\sum_jc_{j,t_{x(n)}}\cdot  \rho_{j,t_{x(n)}}
\end{equation}
Where $\{t_1,t_2,...,t_{x(n)}\}$ is an ordered list containing the $x(n)$ divisors of $n$ with $t_1=1$ and $t_{x(n)}=n$. Hence $\rho_{j,t_i}$ are all partitions with periodicity $t_i$, and $\rho_{t_1}$ is the fully indistinguishable partition state as it has periodicity $t_1=1$. Consequently, in Eq.\ref{eq:2}, $c_{t_1}$ is GI.

We now use Thm.~\ref{thm: uniform_t_probability} to recover the Q-probability $P(Q=i)$ produced by $\rho$.
Noting that all terms with the same periodicity $t_k$ produce the same $P_{t_k}(Q=i)$ and defining $c_{t_i}:=\sum_jc_{j,t_i}$ we have:
\begin{equation}
\label{eq:pqi}
    P(Q=i) = P_{t_1}(Q=i)\cdot c_{t_1} + P_{t_2}(Q=i)\cdot c_{t_2} +...+ P_{t_{x(n)}}(Q=i)\cdot  c_{t_{x(n)}}
\end{equation}

We now define two vectors, one containing the total state output probabilities $\vec{P}= (P(Q=0), P(Q=1), ...,P(Q=n-1))$ whose indices take values in $\{1,\hdots,n\}$ and the vector of periodicities $\vec{c} = (c_1,\hdots,c_n)$ whose indices take values among the $x(n)$ divisors of n.
With this notation Eq.~\ref{eq:pqi} can be cast into a system of $n$ equations in $x(n)$ unknown variables, described by a coefficient matrix $A$:

\begin{equation}\label{system}
    \vec{P} =  A\cdot \vec{c}
\end{equation}

Where the matrix elements of $A$ can be recovered from Thm.~\ref{thm: uniform_t_probability}:

\begin{equation}\label{coef_matrix}
    A_{ij} = P_{j}(Q=i) = \frac{1}{j}\delta_{i\text{ mod } \frac{n}{j}}
\end{equation}

Even though $A \in \mathbb{R}^{n \times x(n)}$ is a rectangular matrix ($x(n) < n$ for $n>2$), we will show that its columns are always linearly independent. This will guarantee that $A$ has full column rank and thus the existence of its Moore–Penrose pseudoinverse $A^+$ \cite{mooreReciprocal1920, penroseGeneralized1955}.
It is therefore always possible to obtain a least-square estimate of $c_{t_1}$ from the experimentally measured output distribution $\vec{P}$ as:

\begin{equation}
    c_{t_1} = \sum_j A^+_{1j} \cdot P(Q=j)
\end{equation}

\begin{proposition}
    $A$ has full column rank and its pseudo inverse $A^+$ takes the form:
    $$
    A^{+}_{ij} = \begin{cases}
     \frac{i}{\varphi(g)} \mu(\frac{g}{i}) &\text{if } gcd(i,g) = i\\
    0 &\text{otherwise}
    \end{cases}
    $$

    Where $g =\frac{n}{gcd(n,j)}$
\end{proposition}
\begin{proof}
    First of all, the role of the factor $g$ can be exposed by rewriting the the delta function of $A$ in terms of divisibility:
    \begin{equation}
        \delta_{i\text{ mod } \frac{n}{j}} = \delta_{gcd(j,\frac{n}{gcd(n,i)}) =\frac{n}{gcd(n,i)}} = \delta_{gcd(j,g') = g'}
    \end{equation} 

    Where, as done in the statement of the proposition we introduced the shorthand $g' = \frac{n}{gcd(n,i)}$
    This relation shows that $A$ can be written as the product of two invertible matrices $D,N$ and a rectangular block diagonal matrix $R$ with full column rank:

\begin{equation}
    A = R\cdot D \cdot N
\end{equation}
\begin{equation}
\begin{split}
    R_{ij} &= \delta_{ j= \frac{n}{gcd(i,n)}} \\
    D_{ij} &= \delta_{gcd(i,j) = i} \\
    N_{ij} &= \frac{1}{j}\cdot \delta_{j,i}
\end{split}
\end{equation}

Where the indeces of the matrices $x(n)$ and $N$ and the column index of $R$ take value among the divisors of $n$ and the row index of $R$ takes values in $\{1,\hdots,n   \}$.
The matrix $N$ can be inverted at sight, while the inverse of the divisibility matrix $D$ can be expressed via the M\"obius function \cite{rotaFoundations1964}:
\begin{equation}
    D^{-1}_{ij} = \mu(\frac{j}{i}) \delta_{gcd(i,j) = i}
\end{equation}
Finally considering $R$, we note that for each $i$, $gcd(i,n)$ takes a unique value and there's thus a unique $j$ such that $gcd(i,n) = \frac{n}{j}$.
Rows of $R$ then never contain more than one value, this ensures linear independence between columns and thus full column rank for $R$ and $A$.
The pseudo-inverse $R^+$ can be computed from the definition \cite{penroseGeneralized1955} and results in:
\begin{equation}
    R^+_{ij} = \frac{1}{\varphi(i)} \delta_{i = \frac{n}{gcd(j,n)}}
\end{equation}

Where now the column index takes values in $\{1,\hdots,n\}$, the row index among the divisors of $n$ and $\varphi(n)$ is Euler's totient function \cite{hardyIntroduction2008}.
The multiplication :
\begin{equation}
\label{eq:pinv_decompose}
N^{-1} \cdot D^{-1}\cdot R^{+} 
\end{equation}
yields the claim.
\end{proof} 

\subsection{Optimality of scaling}\label{sec: optimality}
We have seen that post-selecting on very likely output events of the QFT ($Q\neq0)$ can be used to compute $c_1$ as $P(Q \neq 0)=(1-c_1)(1-\frac{1}{n})$, for many relevant scenarios such as when $n$ is prime or when dealing only with OBB partition states. We will refer to the $1-\frac{1}{n}$ scaling as the \emph{success probability}. 

In this section we prove that this success probability is \emph{optimal}. Indeed, we will show that \emph{any} interferometer and \emph{any} post-selection strategy cannot compute $c_1$ for \emph{any} unknown input state $\rho$ with better than $1-\frac{1}{n}$ success probability.  We note that this optimality result applies to the case when $n$ is prime, and the case where the partition states are all OBB. Note however, that although we do not prove optimality in the most general case ( $\forall$ n, $\forall \rho$), nevertheless in section \ref{sec:numerics} we provide strong numerical evidence that the QFT also massively outperforms the CI for estimating $c_1$ in the general case.

Note that the singly distinguishable state $\rho^s$, where $n-1$ photons are fully indistinguishable and the remaining one is fully orthogonal to the rest, is the closest partition state that we can get to the fully indistinguishable one $\rho^\parallel$. As shown in Stanisic and Turner~\cite{stanisicDiscriminating2018}, this is intuitively the most difficult state to unambiguously discriminate against the fully indistinguishable one. Therefore in the worst-case our unknown input state will be
\begin{equation}\label{eq: construct}
    \rho = c_1 \rho^\parallel + (1-c_1) \rho^s \enspace ,
\end{equation}
we use this worst-case scenario to prove our optimality theorem.

\begin{theorem}[Optimality of our protocol]\label{thm: Optimality}
    Given an initial $n$-photon state $\rho$ of the form of Equation  \ref{eq: density_matrix} where all $\rho_k^{\perp}$ are OBB partition states, the QFT and post-selecting on states such that $Q \neq 0$, is optimal for the task of estimating $c_1$. The same holds if $\rho_k^{\perp}$ are general partition states and $n$ is prime.
\end{theorem}

\begin{proof}
   We proceed with a proof by contradiction.  
   
   Let $\{U', \mathcal{M}'\}$ denote a linear optical interferometer and a set of measurements that can estimate the coefficient $c_1$ for any $\rho$ with an optimal success probability greater than that achieved using the QFT, i.e. $P' > P_{QFT} = 1 - \frac{1}{n}$. In particular,  let $\rho$ be an $n$-photon quantum state constructed as a convex mixture of the fully indistinguishable state $\rho^\parallel$ and a singly distinguishable state $\rho^s$ as in Eq.~\ref{eq: construct}

    As shown in the work of Stanisic and Turner~\cite{stanisicDiscriminating2018}, any set of measurements $\mathcal{M}^{'}$ and any interferometer $U^{'}$ such that $\sum_{M \in \mathcal{M}^{'}}\mathsf{Tr}(MU^{'}(\rho^{\parallel}))=0$, must necessarily have $\sum_{M \in \mathcal{M}^{'}}\mathsf{Tr}(MU^{'}(\rho^{s})) \leq 1-\frac{1}{n}$. This implies that no pair $\{U^{'},\mathcal{M}^{'}\}$ can compute $c_1$ for the state of Eq.~\ref{eq: construct} with a success probability higher than $1-\frac{1}{n}$. However, by assumption we have demanded that $P' > 1-\frac{1}{n}$, which leads to a contradiction. Hence, the pair $\{U', \mathcal{M}'\}$ cannot exist.

\end{proof}

A few closing remarks. Firstly, note that for \emph{some} states $\rho$ one can hope to do better than the $1-\frac{1}{n}$ scaling. Indeed, replacing $\rho^s$ with $\rho^d$ ( a state of fully distinguishable photons) in Eq.~\ref{eq: construct}, one can compute $c_1$ with a scaling of $1-\frac{1}{n!}$~\cite{stanisicDiscriminating2018}. This is however not in contradiction with our optimality result, since in the worst-case, for any unknown input state $\rho$, one cannot do better than the $1-\frac{1}{n}$ scaling. Secondly, note that optimality in the scaling implies optimality in the sample complexity, and below we work this out explicitly.

\section{Protocol scaling}
\label{sec:numerics}

In this section we analyze the (classical and quantum) computational cost of retrieving GI within a fixed error $\epsilon$ with high confidence.
We will compare our protocol to the state-of-the-art \cite{pontQuantifying2022}, and their respective scaling with the number of photons $n$.

\subsection{$n$ prime and OBB sample scaling}\label{sec: prime-scaling}

According to results derived in prior sections, we know that for $n$ prime, GI can be expressed as $c_1 = \frac{(1-1/n)-P(Q\neq0)}{1-1/n}$ (Sec.\ref{sec: OBB}). The same holds for \emph{any n} when all $\rho^{\perp}_j$ are \emph{OBB states} (Sec.\ref{sec: partition_results}).
\\

However, we only have access to a limited number of samples to experimentally estimate $P(Q\neq0)$. Hence, let us denote $\tilde{P}(Q\neq0)$ as the statistical approximation of the output post-selection probability. We can use Hoeffding's inequality \cite{hoeffdingProbability1963} to see that we can estimate $P(Q\neq0)$ to within an additive error $\epsilon$ by performing $O(\frac{1}{\epsilon^2})$ runs.
\\
When we want to estimate GI with a finite number of samples, we have:

\begin{equation}
    \tilde{c_1} = 1-\frac{\tilde{P}(Q\neq0)}{1-1/n}
\end{equation}
Where $\tilde{c_1}$ is our statistical approximation of $c_1$. Now we introduce the error coming from $\tilde{P}(Q\neq0)$:
\begin{equation}
    \tilde{c_1} = 1 -\frac{P(Q\neq0)\pm\epsilon}{1-1/n} = c_1\pm \frac{\epsilon}{1-1/n} =c_1\pm \epsilon'
\end{equation}
Where in the second step we used the original expression for $c_1$ without any statistical error added. Additionally, as we said before, $\epsilon = O(\frac{1}{\sqrt{N_s}})$. Hence the total error of our GI computation scales as:

\begin{equation}
    \epsilon' = O\left(\frac{1}{(1-1/n)\sqrt{N_s}}\right)
\end{equation}
Hence, the sampling cost of estimating $c_1$ up to a fixed error scales with the number of photons as:

\begin{equation}
    N_s = O\left(\frac{1}{(1-1/n)^2}\right) = O\left(\frac{n^2}{(n-1)^2}\right) = O(1)
\end{equation}
\subsection{CI scaling $\forall$ n}
On the other hand, the post-selection probability for the CI scales as

\begin{equation}
    P = \frac{1}{2^{2n-1}}(1+(-1)^n\cdot c_1 \cdot \text{cos}\alpha)
\end{equation}

Using a similar analysis, we see that such success probability expression leads to the following sample complexity:

\begin{equation}
    N_s = O(4^{2n})
\end{equation}

However,combining all $2^n$ different output events with one photon per two modes (see ~\cite{pontQuantifying2022}), the resource scaling can be reduced to:

\begin{equation}
    N_s = O(4^{n})
\end{equation}

Note that the CI scaling \emph{increases exponentially} with the system size as $O(4^n)$, while the sample complexity given by our work is \emph{asymptotically constant} as $O(1)$, independent of $n$.

\subsection{$n$ non-prime sample scaling}\label{sec: non-prime-scale}

As shown in Sec.\ref{sec: non-prime_partition_results}, if $n$ is not a prime number and we are not in the OBB regime, we can always compute $c_1$ by solving the system of equations:
\begin{equation}
    \vec{P} = A \cdot \vec{c} \implies \vec{c} = A^{+} \cdot\vec{P} 
\end{equation}
between the Q-resolved probabilities $\vec{P}= (P(Q=0), P(Q=1), ...,P(Q=n-1))$, and the periodicity coefficients  $\vec{c}= (c_{1},...,c_{n})$. Furthermore in Eq.\ref{eq:pinv_decompose} we showed how $A^+$ can be advantageously decomposed as an invertible part $(DN)^{-1}$ and a pseudo inverse part $R^+$:
\begin{equation}
(A^+)_{ij}=(N^{-1} \cdot D^{-1} \cdot R^{+})_{ij} = \begin{cases}
     \frac{i}{\varphi(g)} \mu(\frac{g}{i}) &\text{if } gcd(i,g) = i\\
    0 &\text{otherwise}
    \end{cases}
\end{equation}

Where $g =\frac{n}{gcd(n,j)}$.From a practical point of view $A^+$ applies the same transformation to all values of $j$ that result in the same $g$, or equivalently it's acting on the vector of probabilities $$P(gcd(Q,n) = k) = \sum_{gcd(j,n)=k}P(Q=j)$$
We can now consider the experimental estimates of these Q marginals $\tilde{P}(gcd(Q,n) = k)$ used in the inversion as independent random variables with average $P(gcd(Q,n) = k)$ and a variance coming from a counting statistics over $N$ repetition of the experiment, which we bound as $\sigma^2 =\frac{1}{N} $.
We can propagate the uncertainty through the application of the reminder of $A^{+}$ applying it once for every $g$ value, obtaining the variance of the experimental estimates $\tilde{c}_i$ as:

\begin{equation}\label{eq:prop_var}
    Var[\tilde{c}_i] = \sigma^2 \cdot\sum_{j|n} (A^{+}_{ij})^2
\end{equation}

For a fixed row i, the non-0 elements of $A^+$ are all the j such that $g = \frac{n}{gcd(n,j)}$ respects $i|g |n$, or equivalently all the divisors of $\frac{n}{i}$ times $i$.
Now given the prime decompositions $n= \prod_k p_k^{l_k}$,$i= \prod_k p_k^{d_k}$, recalling that $\mu$ and $\varphi$ are associative under multiplication, we can decompose then the sum of equation \ref{eq:prop_var} separating the dependence on each prime divisor of $n$.
This manipulation results in:

\begin{equation}
\begin{split}
i^2 \sum_{j|\frac{n}{i}} \frac{\mu^2(j)}{\varphi^2(ij)} &=i^2 \prod_{k}\left( \sum_{m=0}^{l_k}\frac{\mu^2(p_k^{m})}{\varphi^2(p_k^{m+d_k})} \right) = i^2 \prod_k (\frac{1}{\varphi^2(p_k^{d_k})} + \frac{1}{\varphi^2(p_k^{d_k+1})})
\end{split}
\end{equation}

For $i=1$ (computing GI) the summation further simplifies to:
\begin{equation}
    \prod_k\frac{(p_k-1)^2 +1}{(p_k-1)^2} \leq \prod_k\frac{p_k^2}{(p_k-1)^2} = \frac{n^2}{\varphi^2(n)}
\end{equation}
which for $n$ prime evaluates to the inverse of the known optimal success probability factor $(\frac{n}{n-1})^2$. For an arbitrary $n$ the factor still retains a simple formula and can be bounded from above:

\begin{equation}
    \frac{Var[c_1]}{\sigma^2} \leq  \prod_k\frac{p_k^2}{(p_k-1)^2} \leq \prod_k(l_k+1)^2 = x(n)^2
\end{equation}

where x(n) is the number of distinct divisors of $n$. Up to a scaling factor $x(n)^2$ for the variance, $\tilde{c_1}$ behaves exactly as the experimental estimates $\tilde{P}(Q=j)$. We can thus apply Chebyshev's inequality \cite{detchebychefValeurs1867} and recover the number of samples $N_s$ and required to obtain an estimate of $c_1$ precise up to a fixed additive error $\epsilon$ with probability $\delta$:
\begin{equation}
\delta \leq \frac{\mathrm{Var}[c_1]}{\epsilon^2} \leq \frac{x^2(n)}{\epsilon^2 N_s}
\end{equation}
\begin{equation}
N_s \geq \frac{x^2(n)}{\epsilon^2 \delta} = o(\frac{n}{\epsilon^2})
\end{equation}

Additionally, we must account for the classical post-processing complexity that complements the quantum sampling in our protocol. This classical overhead arises from two tasks: first, enumerating all $x(n)$ divisors of $n$, which has a classical complexity of $O(\sqrt{n})$\cite{knuthArt1997}; and second, applying the pseudo-inverse $A^+$ to find $c_1$.

This last operations can actually take advantage of the sparsity of $A^+$.
Through decomposition \ref{eq:pinv_decompose} we can divide the operation in two steps: multiplication by $R^+$ and multiplication by $(D\cdot N)^{-1}$.
$R^+$ is of size $D \times n$ but it's sparse, it has exactly $n$ non-zero entries and can thus be efficiently applied with a linear number of operations.  $(D\cdot N)^{-1}$ isn't sparse but it has size $D \times D$ and thus can be applied in $O(x(n)^2) =o(n)$ operations. 

Overall, the worst scaling part of the post-processing is thus the application of $R^+$, resulting in a final $O(n)$ scaling.

\section{Experimental details}

As discussed in the main text, the experimental results were obtained using the 12-mode reconfigurable universal interferometer available on Quandela Cloud. The details on this machine and the specifications of the single-photon source are available in Ref.~[4] of the main text. At the time of the experiments, the average characterized value of nearest-neighbor pairwise two-photon Hong-Ou-Mandel visibility was $V_\text{avg}=89\pm3\%$. The experiments used photons 2, 3, 4, and 5 (blocking photons 1 and 6). The nearest neighbor pairwise values were $V_{23}=88.5\pm0.5\%$, $V_{34}=91.5\pm0.2\%$, and $V_{45}=86.5\pm0.5\%$. The source also had a multi-photon emission characterized by an integrated two-photon intensity correlation of $g^{(2)}=1.5\pm0.1\%$.

For the QFT protocol, we implemented pseudo-photon-number-resolved (PPNR) detection by spatially dividing the outputs of the QFT into distinct detectors. This spatial division was accomplished using additional QFT unitary transformations on the same reconfigurable chip. Each output of the interference QFT was sent into just one input of another QFT that equally distributes the photons at the output to allow threshold detectors to resolve up to $m$ photons for an $m$-mode QFT. Note that, quantum interference does not occur during this PPNR step since each QFT in this second layer takes just one non-vacuum input that may contain up to $m$ single photons. This means that every photon independently contributes to the detection patterns. However, the photons can still probabilistically bunch at the output of the device, which cannot be resolved by the threshold detectors. This leads to a reduced detection efficiency and a biased probability of not detecting multi-photon events. To compensate for this bias, we post-processed the PPNR distributions to artificially inflate the multi-photon probabilities by the proportion of bunched events that are discarded by the PPNR detection. Notably, this artificial correction is applied after computing the statistical uncertainty of the outcome probability so as not to bias the statistical precision.

The probability of resolving $n$ photons at the output of an $m$-mode PPNR detector is:
\begin{equation}
    p_\text{ppnr}=\binom{m}{n}\frac{n!}{m^n}.
\end{equation}
For the $3$-photon QFT experiment, this means that we compensate by multiplying the probability of $2$-photon pseudo-resolved bunching by $3/2$ and the probability of $3$-photon pseudo-resolved bunching by $9/2$. For the $4$-photon QFT experiment, we compensate by 4/3, 8/3, and 32/3 for 2, 3, and 4 photon outcomes, respectively.

In Figures~\ref{fig:qftcircuits} and \ref{fig:cicircuits} we give the exact Perceval \cite{heurtel2023perceval} circuits used to measure the 3- and 4-photon values of $c_1$. For the 4-photon QFT experiment, we were unable to resolve all outcomes in a single experimental run. Resolving up to 4 photons on each of the 4 output modes of the QFT would require 16 spatial modes (4 modes per PPNR detector). Instead, we performed 4 separate experiments back-to-back, where each experiment resolved up to 4 photons on just 2 of the 4 outputs, while leaving the remaining 2 modes directly monitored by threshold detectors. The distributions from these 4 experiments were then combined in post-processing to produce a full distribution that includes all 4-photon outcomes.

For the cyclic interferometer experiments, we followed the same protocol as in Ref.~\cite{pontQuantifying2022}. We sample the constructive and destructive interference outputs for 21 phases $\phi$ between 0 and $2\pi$, construct the interference fringe, and then fit a sinusoidal curve to extract the amplitude of the fringe. However, we did not correct for any multi-photon errors in any of our experiments.

\begin{figure}
    \centering
    \includegraphics[width=0.435\linewidth]{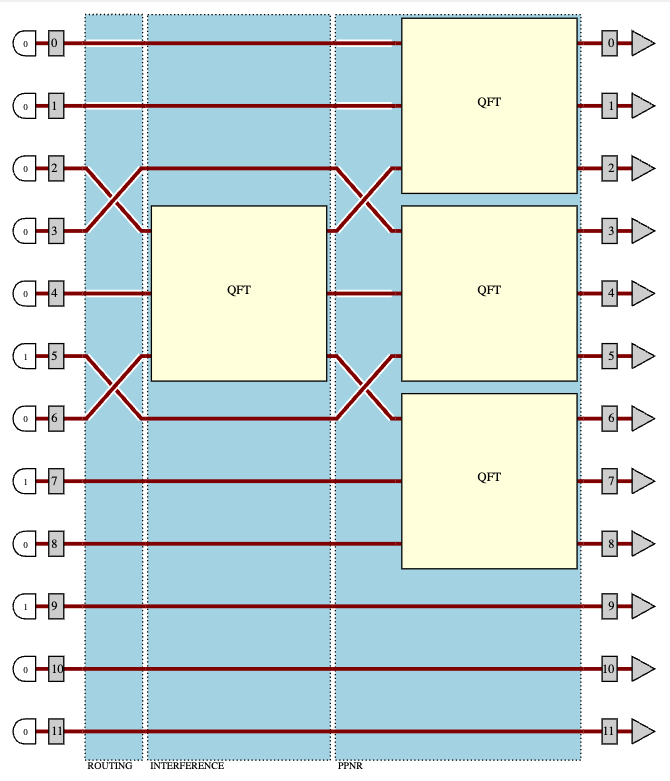}
    \includegraphics[width=0.555\linewidth]{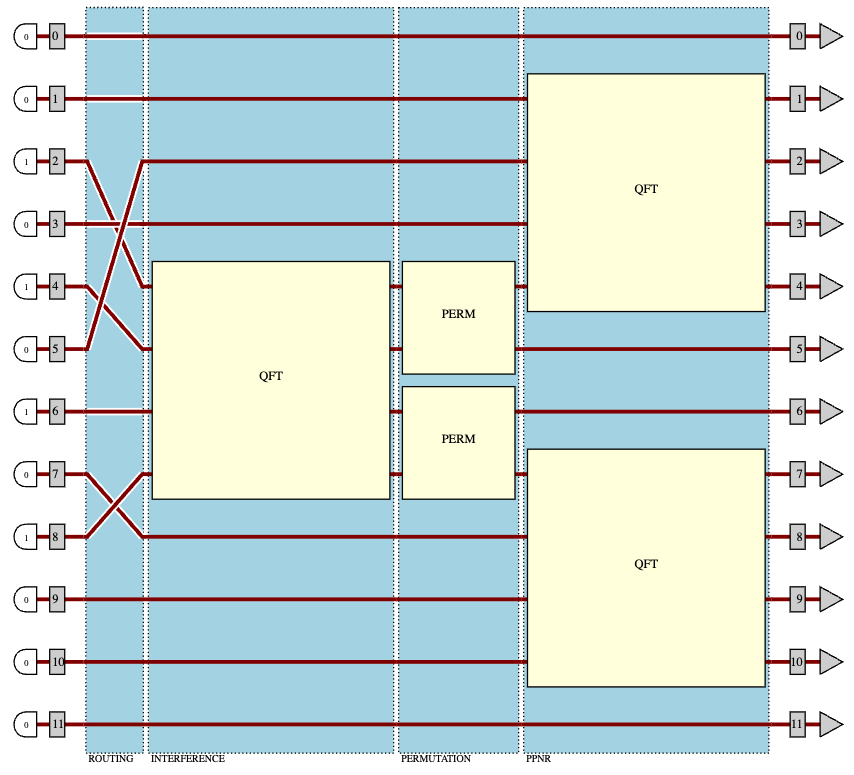}
    \caption{(left) The Perceval circuit to implement the three-photon quantum Fourier transform (QFT) measurement. Photons are input into modes 2, 4, and 6 while all other modes are vacuum. Routing is done to bring all three photons into the first QFT and then separate the outputs into three independent QFTs that implement the PPNR detection. (right) The Perceval circuit to implement the four-photon QFT measurement. The output modes of the interference layer are routed either to a four-photon PPNR QFT or sent directly to a threshold detector. The PERM operation is either an identity or a swap, giving rise to four circuits that are recombined to obtain the full 4-photon resolved distribution.}
    \label{fig:qftcircuits}
\end{figure}

\begin{figure}
    \centering
    \includegraphics[width=0.39\linewidth]{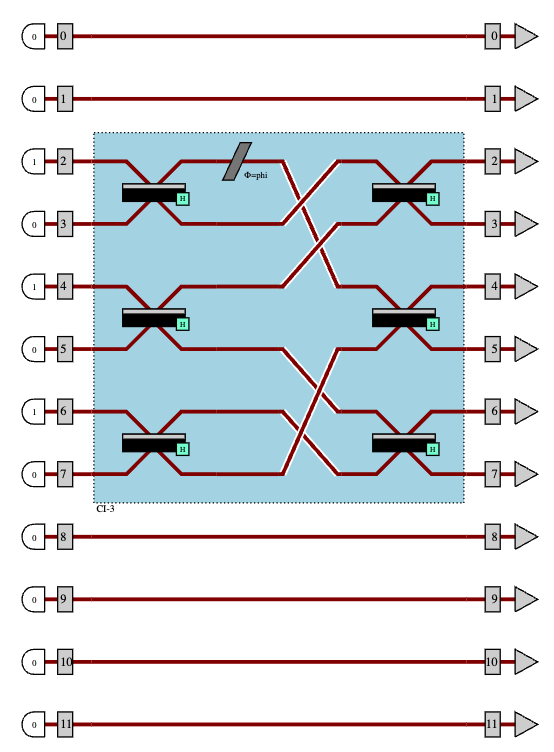}
    \includegraphics[width=0.47\linewidth]{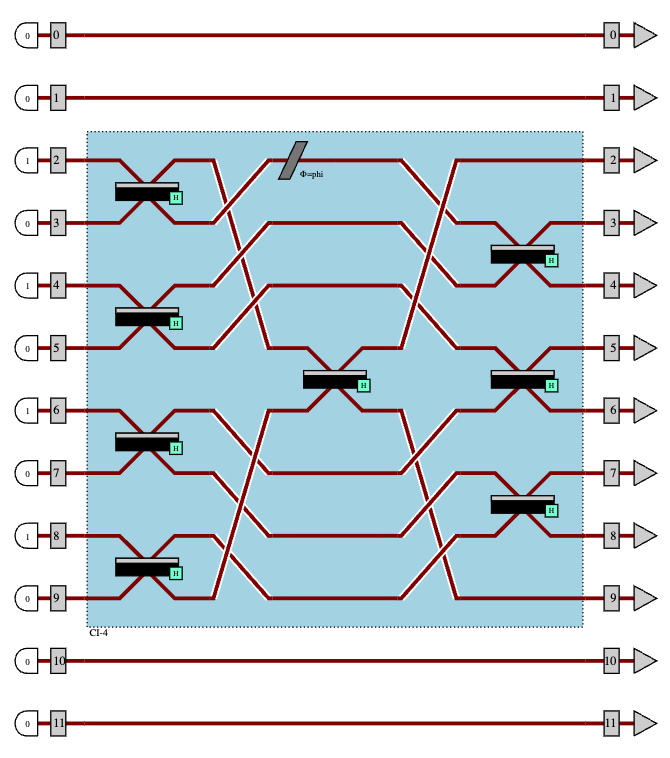}
    \caption{(left) The Perceval circuit to implement the three-photon cyclic interferometer (CI) measurement. Photons are input into modes 2, 4, and 6 while all other modes are vacuum. (right) The Perceval circuit to implement the four-photon CI measurement. The photons are input into modes 2, 4, 6, and 8 while all other modes are vacuum. In both cases, the phase $\phi$ is varied to obtain the interference fringe needed to extract $c_1$.}
    \label{fig:cicircuits}
\end{figure}
\section{Uniformity law for the periodic ZTL}

Our prior results all apply to a case where we have an $n$-photon input state, where $n=m$ with one photon per mode. This is because, for our protocol, we require that for the fully indistinguishable state, $P_i(Q\neq0) = 0$ and according to the \ZTL{} this is only true in this setting. 

We now present a general property of the \pZTL{} which we derived while working on our results and which we believe can have applications beyond our work. In particular, we prove the uniformity of the probability distribution for the non-suppressed events dictated by the \pZTL{} when we have $n\neq m$ photons. In order to prove such uniformity, we introduce the following theorem:

\begin{theorem}\label{thm: proof_of_uniformity}
    Given a $t$-periodic Fock state $\vec{r}_t$ (see Definition~\ref{defn: t-periodicFock}) with $n_t$ photons per period, in $m$ modes, subjected to $U = QFT_m$. For output states where $Q(\vec{s})  = 0 \mod \frac{m}{t}$, each $Q$ value has an equal probability of $1/t$, if $gcd(n_t, t) = 1$.
\end{theorem}
\begin{proof}
In order to prove this property, we use the shift-invariance of the $QFT$ for linear optics, as derived in Sec.~\ref{sec: shift_invariance}. 

We define an $m$-mode output vector as $\vec{s} = (s_1, ..., s_m)$ where $s_j \in (1, ..., n)$. Now, let $\vec{s} \ '$ be another output state whose elements are constructed by shifting the elements of $\vec{s}$ to the following mode in a periodic way (s.t. the elements in the mode $m$ get shifted to mode $1$) so that we have: $s_j' = s_{j+1 \ \text{mod} \ m}$. We re-write $Q(\vec{s})$ as: 

\begin{equation}
    Q(\vec{s}) = \sum\limits_{j=0}^{m-1} s_j \cdot j \ \ \text{mod} \ m \enspace , 
\end{equation}
where we have just re-written the sum of the elements of the mode-assignment vector as $\sum\limits_i \vec{d}_i(s) = \sum\limits_{j=0}^{m-1} s_j \cdot j$. We can now write $Q(\vec{s} \ ')$ in terms of $Q(\vec{s})$ as: 
\begin{equation}
    \begin{split}
        Q(\vec{s} \ ') &= \left(\sum_{j=0}^{m-1} s_j'\cdot j \right) \ \text{mod} \ m = \left(\sum_{j=0}^{m-1} s_{j+1}\cdot j \right)\ \text{mod} \ m= \left(\sum_{k=0}^{m-1} s_{k}\cdot (k-1)\right) \ \text{mod} \ m \\
        &= \left(\sum\limits_{k=0}^{m-1} s_k \cdot k \right) \ \text{mod} \ m \  - \left(\sum\limits_{k=0}^{m-1} s_k \right) \ \text{mod} \ m  = Q(\vec{s}) - n \ \text{mod} \ m \\
        &= Q(\vec{s}) - \left(\frac{m}{t}\cdot n_t\right) \ \text{mod} \ m \enspace ,
    \end{split}
\end{equation}
re-calling that $n_t$ is the number of photons per period $t$ and therefore the total number of photons is the number of repeated patterns $\frac{m}{t}$ multiplied by $n_t$. The subset of states generated by shifting the elements of vector $\vec{s}$ are related, such that each shift subtracts a multiple of $\frac{m}{t}\cdot n_t$ from $Q(\vec{s})$:
\begin{equation}\label{eq: Q_multiples}
    \begin{split}
        Q(\vec{s}\ ') &= Q(\vec{s}) - \left(\frac{m}{t}\cdot n_t\right) \ \text{mod} \ m \\
        Q(\vec{s}\ '') &= Q(\vec{s}) - \left(2\cdot\frac{m}{t}\cdot n_t\right) \ \text{mod} \ m \\
        Q(\vec{s}\ ^{m-1}) &= Q(\vec{s}) - \left((m-1)\cdot\frac{m}{t}\cdot n_t\right) \ \text{mod} \ m \\ \enspace .
    \end{split}
\end{equation}
    
Since we are working with quantities $\text{mod} \ m$, a direct symmetry is that $Q(\vec{s} \ ^{(i)}) = Q(\vec{s} \ ^{(i+t) \ \text{mod} \ m}), \ \forall \ n_t$.

What the above tells us is that if $Q(\vec{s})$ is a multiple of $\frac{m}{t}$, the shifted $Q(\vec{s}^{\ (i)})$ are also multiples of $\frac{m}{t}$ and therefore any non-suppressed outputs leading to $Q$ mean all the shifted outputs in that family are also non-suppressed. This applies to all outputs and their shifted families because these properties are independent of the specific output $\vec{s}$. Therefore, if the $Q$ values are uniformly distributed among a general shifted family of outputs, this means that they are uniformly distributed across \emph{all} shifted families

    Due to the shift-invariance of the $QFT$, applying the interferometer $U = QFT_m$ means that output states equivalent up to a global shift will share the same transition probability, $P_{\vec{r} \rightarrow \vec{s}} = P_{\vec{r} \rightarrow \vec{s}^{(i)}} \  \forall \ (i = 0, ..., m-1)$. Therefore, each possible output in the shifted families have the same probability which, if the $Q$ values are uniformly distributed, then  leads to an equal probability for each value $Q$ (since they appear an equal number of times).

    The task then becomes to show \emph{when} $Q$ is uniformly distributed in a general shifted family. Let us first re-write the last line of Eq.~\ref{eq: Q_multiples} as:
    \begin{equation}\label{eq: Q_shift}
        Q^{l} = k \cdot \frac{m}{t} - \left(l\cdot\frac{m}{t}\cdot n_t\right) \ \text{mod} \ m \enspace ,
    \end{equation}
    where we have replaced $Q(\vec{s})$ with a value that is not suppressed (by the periodic ZTL) $k\cdot \frac{m}{t}$ and the shifted $Q$ is labelled by an arbitrary shift $l$. It is important to note that any periodic state will necessarily have $t \leq m$ and $\frac{m}{t}$ an integer, and that symmetry dictates $Q(\vec{s} \ ^{(i)}) = Q(\vec{s} \ ^{(i+t) \ \text{mod} \ m}), \ \forall \ n_t$, i.e. that any output vectors with value $Q$, if shifted by period $t$, give an equal value of $Q$. This means that we only have to show that for shifts $l = (0, ..., t-1)$ each value of $Q$ appears once. If this is the case, then this pattern will be repeated up to $m-1$ shifts, leading to uniformly distributed $Q$ values. Since $k$ is fixed, we only need to focus on the second term of Eq.~\ref{eq: Q_shift} to show that $Q^l$ will reach every value of $Q$. In order to prove it, we start from the division identity with $x = \frac{m \cdot n_t\cdot l}{t}$:
    \begin{equation}
        \frac{x}{m} = \frac{x \bmod m}{m} + a \implies \frac{n_t l}{t} = \frac{(\frac{m \cdot n_t\cdot l}{t} \bmod m) \cdot t}{m \cdot t} + a = \frac{(\frac{m \cdot n_t\cdot l}{t} \bmod m) \cdot t/m}{t} + a
    \end{equation}
    Where we have multiplied the first term in the right-hand side by ($\frac{t}{t}$). Applying the division identity again with $x = n_t \cdot l$ and modulus $t$:
    \begin{equation}
        \frac{n_t \cdot l}{t} = \frac{(n_t \cdot l \bmod t)}{t} + a 
    \end{equation}
   Equating the RHS of the two previous equations, we finally obtain:
   \begin{equation}
       (\frac{m \cdot n_t\cdot l}{t} \bmod m) \cdot  = \frac{m}{t}\cdot (n_t \cdot l \bmod t)
   \end{equation}
   Now the task reduces to showing that $(n_t \cdot l \bmod t)$ generates all values modulo $t$ exactly once when varying $l$. We know that $l = \mathbb{Z}_t := {0, ..., t-1}$ the set of all values modulo $t$. We can define a function $f: \mathbb{Z}_t \mapsto \mathbb{Z}_t$ that acts on $l$ as $f(l) = l \cdot n_t \ \text{mod} \ t$. In order to show that this function maps distinct values of $\mathbb{Z}_t$ to other distinct values in the set, we have to prove that the function $f$ is bijective, such that $n_t \ \text{mod} \ t$ simply permutes all values in the set $\mathbb{Z}_t$. We do this by first showing that $f$ is injective, i.e. it is a one-to-one function. In order for $f$ to be injective, we must prove that $f(l_1) = f(l_2) \implies l_1 = l_2$. From modular arithmetic, we know that:
   \begin{equation}
       k\cdot a \equiv k\cdot b \text{ mod }p, \text{ }gcd(k,p)=1 \implies  a \equiv  b \text{ mod }p
   \end{equation}
    Which directly implies in our formalism $f(l_1) = f(l_2) \implies l_1 = l_2$, proving injectivity for $gcd(n_t, t) = 1$ (i.e. when $n_t$ and $t$ are co-prime). We have shown that $f(l)$ is injective, and since $l \in \mathbb{Z}_t $, i.e. we have a finite set mapping to a finite set, then we have bijectivity. Therefore $Q^l$ spans all values of $Q$ exactly once, when $n_t$ and $t$ are co-prime. If $n$ is prime then this is always true, since the period $t=n$ and $n_t < n$.
\end{proof}


\end{document}